\documentclass[peerreview,a4paper,12pt]{IEEEtran}

\usepackage{amsthm}
\usepackage{amsmath}
\usepackage{amsfonts,amssymb,amsmath}
\usepackage{graphicx}
\usepackage{epsfig}
\usepackage[applemac]{inputenc}
\usepackage{latexsym}

\newtheorem{mycor}{Corollary}
\newtheorem{myprop}{Proposition}
\newtheorem{mythe}{Theorem}
\newtheorem{mylem}{Lemma}

\newtheorem{myrem}{Remark}

\DeclareMathOperator*{\Tr}{Tr}

\begin{document}

\sloppy

\title{Polar Codes for Arbitrary Classical-Quantum Channels and Arbitrary cq-MACs}


%
%

\author{Rajai Nasser, \IEEEmembership{Member, IEEE,} and Joseph M. Renes, \IEEEmembership{Member, IEEE}%
}




\maketitle

\begin{abstract}
We prove polarization theorems for arbitrary classical-quantum (cq) channels. The input alphabet is endowed with an arbitrary Abelian group operation and an Ar{\i}kan-style transformation is applied using this operation. It is shown that as the number of polarization steps becomes large, the synthetic cq-channels polarize to deterministic homomorphism channels which project their input to a quotient group of the input alphabet. This result is used to construct polar codes for arbitrary cq-channels and arbitrary classical-quantum multiple access channels (cq-MAC). The encoder can be implemented in $O(N\log N)$ operations, where $N$ is the blocklength of the code. A quantum successive cancellation decoder for the constructed codes is proposed. It is shown that the probability of error of this decoder decays faster than $2^{-N^{\beta}}$ for any $\beta<\frac{1}{2}$.
\end{abstract}

\section{Introduction}

Polar coding is the first efficient coding technique that was shown to achieve the capacity of symmetric binary-input channels \cite{Arikan}. The code construction relies on a phenomenon called polarization: starting from a collection of independent copies of a given binary-input channel, one can recursively apply a polarization transformation on those channels and obtain synthetic channels that become extreme (i.e., either almost useless or almost perfect channels) as the number of polarization steps becomes large. This suggests sending information through the channels which are almost perfect, while sending frozen symbols through the almost useless channels. Since the total capacity is conserved by the applied transformations, we can reliably communicate using this method at a rate that is close to the capacity. Ar{\i}kan proposed a successive cancellation decoder for the constructed polar code and he showed that both the encoder and the decoder can be implemented in $O(N\log N)$ operations. The probability of error of the successive cancellation decoder was shown to decay faster than $2^{-N^\beta}$ for any $\beta<\frac{1}{2}$ \cite{ArikanTelatar}.

Since Ar{\i}kan's polarization transformation for binary-input channels uses the XOR operation, the straightforward generalization of Ar{\i}kan's construction to arbitrary discrete memoryless channel is to replace the XOR operation with a binary operation on the input alphabet. It was shown that polarization happens for a wide family of binary operations: addition modulo $q$ (where $q$ is prime) \cite{SasogluTelAri}, addition modulo $2^r$ \cite{ParkBarg}, arbitrary Abelian group operations \cite{SahebiPradhan} and arbitrary quasigroup operations \cite{RajTelA}. This allowed the construction of polar codes for arbitrary discrete memoryless channels since any set can be endowed with an Abelian group operation. Note that in the case where the input alphabet size is not prime, the polarization may not be a two-level polarization to useless and perfect channels as in the binary-input case. We may have multilevel polarization where it is possible for the synthetic channels to converge to intermediate channels which are neither almost useless nor almost perfect. However, the polarized intermediate channels are ``easy" in the sense that it is easy to reliably communicate information through them at a rate that is near their symmetric capacity. A complete characterization of binary operations which are polarizing was given in \cite{RajErgI} and \cite{RajErgII}.

Polarization was also shown to happen in the multiple access setting. Polar codes were constructed for two-user MACs with inputs in $\mathbb{F}_q$ \cite{SasogluTelYeh}, for $m$-user binary-input MACs \cite{AbbeTelatar}, and for arbitrary MACs \cite{RajTelA, RajErgII}.

Wilde and Guha constructed polar codes for binary-input classical-quantum channels in \cite{WildeGuhaCQ}. They showed that using the same polarization transformation of Ar{\i}kan yields polarization of the synthetic cq-channels to almost useless and almost perfect channels. Wilde and Guha proposed a quantum successive cancellation decoder and showed that its probability of error decays faster than $2^{-N^\beta}$ for any $\beta<\frac{1}{2}$. In \cite{PolarcqMAC}, Hirche et. al. constructed codes for binary-input cq-MAC codes by combining the polarization results of \cite{WildeGuhaCQ} with the monotone chain rule method of \cite{Monotone}.

In this paper, we construct polar codes for arbitrary cq-channels and arbitrary cq-MACs by using arbitrary Abelian group operations on the input alphabets. The polarization transformation that we use is similar to the one in \cite{SahebiPradhan}. Since we are proving a quantum version of the results in \cite{SasogluTelAri} and \cite{SahebiPradhan}, many ideas of those two papers were adopted and adapted to the quantum setting. However, some inequalities that were used in \cite{SasogluTelAri} and \cite{SahebiPradhan} do not have quantum analogues. Therefore, other inequalities that serve the same purpose have to be shown for cq-channels.

In section II we give useful definitions and basic results that we will use later. The polarization transformation is described in section III. Two-level polarization is shown in section IV for cq-channels having input in $\mathbb{F}_q$. In Section V, we prove multilevel polarization for arbitrary cq-channels using an arbitrary Abelian group operation on the input alphabet. We show that the synthetic cq-channels converge to deterministic homomorphism channels which project their input onto a quotient group of the input alphabet. The rate of polarization is discussed in section VI. Polar codes are constructed and studied in section VII. As in all polar coding schemes, the encoder can be implemented in $O(N\log N)$ operations, where $N$ is the blocklength of the polar code. We prove that the probability of error of the quantum successive cancellation decoder decays faster than $2^{-N^\beta}$ for any $\beta<\frac{1}{2}$, but we do not have an efficient implementation of the decoder. Finally, we discuss polarization of arbitrary cq-MACs in section VIII. We show that while cq-MAC polar codes may not achieve the whole symmetric capacity region, they always achieve points on the dominant face. We show that the whole symmetric capacity region can be achieved by combining our cq-channel polarization result with the rate-splitting method of \cite{SasogluTelYeh} or with the monotone chain rule method of \cite{Monotone}.

\section{Preliminaries}

A classical-quantum (cq) channel $W: x\in G\longrightarrow \rho_x \in \mathcal{DM}(k)$ takes a classical input $x\in G$ and has a quantum output $\rho_x \in \mathcal{DM}(k)$, where $\mathcal{DM}(k)$ is the space of density matrices of dimension $k<\infty$. We assume that the input alphabet $G$ is finite but its size $q=|G|$ can be arbitrary.

If the input to the cq-channel $W$ is uniformly distributed, we can describe the state of the joint input-output system as the state $\rho^{XB}\in \mathcal{DM}(q\cdot k)$ defined as:
$$\rho^{XB}:=\frac{1}{q}\sum_{x\in G}|x\rangle\langle x|\otimes \rho_x.$$

A very important quantity associated with $W$ is the symmetric Holevo information $I(W)$ defined as:
$$I(W):=I(X;B)_\rho=H(X)_\rho+H(B)_\rho-H(XB)_\rho,$$

where $H(\sigma)$ is the von Neumann entropy of the density matrix $\sigma$:
$$H(\sigma)=-\Tr (\sigma\log\sigma),$$
and $\log$ is the natural logarithm operator. It is easy to show that

$$I(W)=H\left(\frac{1}{q}\sum_{x\in G}\rho_x\right)-\frac{1}{q}\sum_{x\in G}H(\rho_x).$$

The quantity $I(W)$ is the capacity for transmitting classical information over the channel $W$ when the prior input distribution is restricted to be uniform in $G$. We have $0\leq I(W)\leq \log q$.

Besides $I(W)$, we will need another parameter that measures the reliability of the channel $W$. For the binary-input case, the fidelity between the two output states was used as a measure of reliability in \cite{WildeGuhaCQ}. In our case, we have $q$ output states, so we will consider the average pairwise fidelity between them (similarly to the average Bhattacharyya distance defined in \cite{SasogluTelAri}):
$$F(W):=\frac{1}{q(q-1)}\sum_{\substack{x,x'\in G,\\x\neq x'}}F(\rho_x,\rho_{x'}),$$

where $F(\rho,\sigma)=\Tr\sqrt{\rho^{\frac{1}{2}}\sigma\rho^{\frac{1}{2}}}=\left\|\sqrt{\sigma}\sqrt{\rho}\right\|_1$, and $\|A\|_1$ is the nuclear norm of the matrix $A$:
$$\|A\|_1=\Tr\sqrt{A^\dagger A}.$$
Clearly, $0\leq F(W)\leq 1$. We adopt the convention $F(W):=0$ if $|G|=1$.

It was was shown in \cite{AvFidelity} that $\mathbb{P}_e(W)\leq (q-1)F(W)$, where $\mathbb{P}_e(W)$ is the probability of error of the optimal decoder of $W$. This shows that if $F(W)$ is small then $\mathbb{P}_e(W)$ is also small and so $W$ is reliable. Intuitively, this is true because a small $F(W)$ means that all the pairwise fidelities are small, which implies that all the output states are easily distinguishable from each other, which in turn should allow a reliable decoding.

The following proposition provides three inequalities that relate $I(W)$ and $F(W)$.

\begin{myprop}
\label{propFid}
We have:
\begin{itemize}
\item[(i)] $\displaystyle I(W)\geq \log \frac{q}{1+(q-1)F(W)}$.
\item[(ii)] $I(W)\leq \log(q/2) + (\log 2)\sqrt{1-F(W)^2}$.
\item[(iii)] $I(W)\leq \log\left(1+\sqrt{q^2-(1+(q-1)F(W))^2}\right)$.
\end{itemize}
\end{myprop}
\begin{proof}
See Appendix \ref{appFid}.
\end{proof}

\vspace*{3mm}

In the above proposition, the first inequality implies that if $I(W)$ is close to 0 then $F(W)$ is close to 1. The same inequality also implies that if $F(W)$ is close to 0 then $I(W)$ is close to $\log q$. The second inequality implies that if $I(W)$ is close to $\log q$ then $F(W)$ is close to 0. The third inequality implies that if $F(W)$ is close to $1$ then $I(W)$ is close to 0.

\subsection{Non-commutative union bound}
Sen proved in \cite{Sen} the following ``non-commutative union bound":
\begin{equation}
\label{eqNonComUnionBoundSen}
1-\Tr(\Pi_r\ldots\Pi_1\rho\Pi_1\ldots\Pi_r)\leq 2\sqrt{\sum_{i=1}^r(1-\Tr(\Pi_i\rho))},
\end{equation}
where $\Pi_1,\ldots,\Pi_r$ are projection operators. This inequality was used in \cite{WildeGuhaCQ} to upper bound the probability of error of the quantum-successive cancellation decoder of the polar code constructed for a binary-input cq-channel. This was possible because the measurements used in \cite{WildeGuhaCQ} are projective. In this paper, the quantum successive cancellation decoder that we propose uses general POVM measurement. Therefore, we cannot use the inequality \eqref{eqNonComUnionBoundSen}.

We provide a ``non-commutative union bound" that is looser than \eqref{eqNonComUnionBoundSen} by a multiplicative factor of $\sqrt{r}$, but it is more general so that it can be applied to general POVMs.
\begin{mylem}
\label{lemNonComUnionBound}
Let $\Pi_1,\ldots,\Pi_r$ be $r$ positive operators satisfying $\Pi_1\leq I,\ldots,\Pi_r\leq I$. We have:
$$ 1-\Tr\left(\sqrt{\Pi_r}\ldots\sqrt{\Pi_1}\rho\sqrt{\Pi_1}\ldots\sqrt{\Pi_r}\right)\leq 2\sqrt{r}\sqrt{\sum_{i=1}^r\left(1-\Tr(\Pi_i\rho)\right)}.$$
\end{mylem}
\begin{proof}
See Appendix \ref{appNonComUnionBound}.
\end{proof}

\section{Polarization process}

Since any set can be endowed with an Abelian group operation, we may assume that one such operation on $G$ is fixed. We will denote this Abelian group operation additively.

Let $W: x\in G\longrightarrow \rho_x \in \mathcal{DM}(k)$ be a cq-channel. Define the channels $W^-: u_1\in G\longrightarrow \rho_{u_1}^- \in \mathcal{DM}(k^2)$ and $W^+: u_2\in G\longrightarrow \rho_{u_2}^+ \in \mathcal{DM}(k^2\cdot q)$ as:
$$\rho_{u_1}^-=\frac{1}{q}\sum_{u_2\in G} \rho_{u_1+u_2}\otimes\rho_{u_2},$$
and
$$\rho_{u_2}^+=\frac{1}{q}\sum_{u_1\in G} \rho_{u_1+u_2}\otimes\rho_{u_2}\otimes|u_1\rangle\langle u_1|.$$

Moreover for every $n>0$ and every $s=(s_1,\ldots,s_n)\in\{-,+\}^n$, define $W^s=(\ldots((W^{s_1})^{s_2})\ldots)^{s_n}$.

\begin{myrem}
\label{remPolarSim}
$W^-$ and $W^+$ can be constructed as follows:
\begin{itemize}
\item Two independent and uniform random variables $U_1,U_2$ are generated in $G$.
\item $X_1=U_1+U_2$ and $X_2=U_2$ are computed.
\item $X_1$ is sent through one copy of the channel $W$. Let $B_1$ be the quantum system describing the output.
\item $X_2$ is sent through another copy of the channel $W$ (independent from the one that was used for $X_1$). Let $B_2$ be the quantum system describing the output.
\end{itemize}
It can be easily seen that the channels $U_1\longrightarrow B_1B_2$ and $U_2\longrightarrow B_1B_2U_1$ simulate $W^-$ and $W^+$ respectively. 
\end{myrem}

We have:
\begin{align*}
I(W^-)+I(W^+)&=I(U_1;B_1B_2)+I(U_2;B_1B_2U_1)=I(U_1;B_1B_2)+I(U_2;B_1B_2|U_1)\\
&=I(U_1U_2;B_1B_2)=I(X_1X_2;B_1B_2)=I(X_1;B_1)+I(X_2;B_2)=2I(W).
\end{align*}

This shows that the total symmetric Holevo information is conserved. Moreover,
$$I(W^+)=I(U_2;B_1B_2U_1)\geq I(U_2;B_2)=I(X_2;B_2)=I(W)$$
and
$$I(W^-)=2I(W)-I(W^+)\leq I(W).$$

Let us now study the reliability of the channel and how it is affected after one step of polarization. But first let us define the quantity $F_d(W)$ for every $d\in G$:
$$F_d(W)=\frac{1}{q}\sum_{x\in G}F(\rho_x,\rho_{x+d}).$$
Clearly, $0\leq F_d(W)\leq 1$ and $F_0(W)=1$. Note that
$$F(W)=\frac{1}{q-1}\sum_{\substack{d\in G,\\d\neq 0}}F_d(W).$$

Define $\displaystyle F_{\max}(W)=\max_{\substack{d\in G,\\d\neq 0}}F_d(W)$. Clearly, $F(W)\leq F_{\max}(W)\leq (q-1)F(W)$.

\begin{myprop}
\label{propFPolar}
For every $d\in G$, we have:
\begin{itemize}
\item $F_d(W^+)=F_d(W)^2$.
\item $\displaystyle F_d(W)\leq F_d(W^-)\leq 2F_d(W) + \sum_{\substack{\Delta\in G,\\
\Delta\neq 0,\\\Delta\neq -d}}F_\Delta(W)F_{d+\Delta}(W)$.
\end{itemize}
\end{myprop}
\begin{proof}
See Appendix \ref{appFPolar}.
\end{proof}

\begin{mycor}
\label{corFidPolar}
We have:
\begin{itemize}
\item $F_{\max}(W^+)=F_{\max}(W)^2$.
\item $F_{\max}(W)\leq F_{\max}(W^-)\leq q F_{\max}(W)$.
\item $F(W^+)\leq \min\Big\{F(W),\;(q-1)^2 F(W)^2\Big\}$.
\item $F(W)\leq F(W^-)\leq  q(q-1)F(W)$
\end{itemize}
\end{mycor}
\begin{proof}
First equation:
$$F_{\max}(W^+)=\max_{\substack{d\in G,\\d\neq 0}} F_d(W^+)=\max_{\substack{d\in G,\\d\neq 0}} F_d(W)^2=\left(\max_{\substack{d\in G,\\d\neq 0}} F_d(W)\right)^2=F_{\max}(W)^2.$$

\vspace*{3mm}

Second equation:
\begin{align*}
F_{\max}(W)&=\max_{\substack{d\in G,\\d\neq 0}} F_d(W)\leq \max_{\substack{d\in G,\\d\neq 0}} F_d(W^-)=F_{\max}(W^-)\\
&\leq \max_{\substack{d\in G,\\d\neq 0}} \Big(2 F_d(W) + \sum_{\substack{\Delta\in G,\\
\Delta\neq 0,\\\Delta\neq -d}}F_\Delta(W)F_{d+\Delta}(W)\Big)\\
&\leq 2 F_{\max}(W) + (q-2)F_{\max}(W)^2\leq q F_{\max}(W).
\end{align*}

\vspace*{3mm}

First part of third equation:
$$F(W^+)=\frac{1}{q-1}\sum_{\substack{d\in G,\\d\neq 0}}F_d(W^+)=\frac{1}{q-1}\sum_{\substack{d\in G,\\d\neq 0}}F_d(W)^2\leq \frac{1}{q-1}\sum_{\substack{d\in G,\\d\neq 0}}F_d(W)=F(W).$$

\vspace*{3mm}

Second part of third equation:
$$F(W^+)\leq F_{\max}(W^+)=F_{\max}(W)^2\leq (q-1)^2F(W)^2.$$

\vspace*{3mm}

First inequality of the fourth equation:
$$F(W^-)=\frac{1}{q-1}\sum_{\substack{d\in G,\\d\neq 0}}F_d(W^-)\geq \frac{1}{q-1}\sum_{\substack{d\in G,\\d\neq 0}}F_d(W)=F(W).$$

\vspace*{3mm}

Second inequality of the fourth equation:
$$F(W^-)\leq F_{\max}(W^-)\leq q F_{\max}(W)\leq q(q-1)F(W).$$
\end{proof}

The following lemma is very useful to prove polarization results.

\begin{mylem}
\cite{SahebiPradhan}
\label{lemPolarIT}
Let $\{B_n\}_{n\geq 0}$ be a sequence of independent and uniformly distributed $\{-,+\}$-valued random variables. Suppose $\{I_n\}_{n\geq 0}$ and $\{T_n\}_{n\geq 0}$ are two processes adapted to the process $\{B_n\}_{n\geq 0}$ satisfying:
\begin{itemize}
\item[(1)] $0\leq I_n\leq \log q$.
\item[(2)] $\{I_n\}_{n\geq 0}$ converges almost surely to a random variable $I_\infty$.
\item[(3)] $0\leq T_n\leq 1$.
\item[(4)] $T_{n+1}=T_n^2$ when $B_{n+1}=+$.
\item[(5)] There exists a function $f(\epsilon)$ (depending only on $q$) satisfying $\displaystyle\lim_{\epsilon\to0}f(\epsilon)=0$ such that for all $n$, if $T_n<\epsilon$ then $I_n>\log q - f(\epsilon)$.
\item[(6)] There exists a function $g(\epsilon)$ (depending only on $q$) satisfying $\displaystyle\lim_{\epsilon\to0}g(\epsilon)=0$ such that for all $n$, if $T_n>1-\epsilon$ then $I_n<g(\epsilon)$.
\end{itemize}
Then $\displaystyle T_{\infty}=\lim_{n\to\infty}T_n$ exists almost surely. Moreover, we have $I_{\infty}\in \{0,\log q\}$ and $T_{\infty}\in \{0,1\}$ with probability 1.
\end{mylem}

\section{Polarization for $G=\mathbb{F}_q$}

In this section, we focus on the particular case where $G=\mathbb{F}_q$ where $q$ is prime. The main result of this section is the following theorem.

\begin{mythe}
Let $W: x\in \mathbb{F}_q\longrightarrow \rho_x \in \mathcal{DM}(k)$ be a cq-channel with input in $\mathbb{F}_q$. For every $\delta>0$, we have:
\begin{equation}
\lim_{n\to\infty} \frac{1}{2^n}\Big|\Big\{s\in\{-,+\}^n:\; \delta \leq I(W^s) \leq \log q -\delta\Big\}\Big|=0.
\label{eqPolar}
\end{equation}
Moreover, for every $\beta<\frac{1}{2}$, we have:
\begin{equation}
\lim_{n\to\infty} \frac{1}{2^n}\left|\left\{s\in\{-,+\}^n:\; I(W^s)\geq \log q - \delta,\; F(W^s)<2^{-2^{\beta n}}\right\}\right|=\frac{1}{\log q}I(W).
\label{eqPolarRate}
\end{equation}
\label{thePolarFq}
\end{mythe}
\begin{proof}
Let $\{B_n\}_{n\geq 0}$ be a sequence of independent and uniformly distributed $\{-,+\}$-valued random variables. Define the cq-channel-valued process $\{W_n\}_{n\geq 0}$ as follows:
\begin{itemize}
\item $W_0=W$.
\item $W_{n}=W_{n-1}^{B_n}$ for every $n\geq 1$.
\end{itemize}
Let $I_n=I(W_n)$ and $\displaystyle T_n=F_{\max}(W_n)$. Let us check the conditions of Lemma \ref{lemPolarIT}. Conditions (1) and (3) follow from the properties of $I(W)$ and $F_{\max}(W)$. Condition (4) is satisfied because of Corollary \ref{corFidPolar}.

We have $\displaystyle\mathbb{E}(I_{n+1}|W_n)=\frac{1}{2}I(W_n^-)+\frac{1}{2}I(W_n^+)=I(W_n)$. This shows that $\{I_n\}_{n\geq 0}$ is a bounded martingale and so it converges almost surely. This shows that condition (2) is satisfied.

Condition (5) follows from the following inequality:
\begin{align*}
I(W)\stackrel{(a)}{\geq} \log\frac{q}{1+(q-1)F(W)}\geq \log\frac{q}{1+(q-1)F_{\max}(W)},
\end{align*}
where (a) is from Proposition \ref{propFid}. By choosing $f(\epsilon)=\log(1+(q-1)\epsilon)$, we can see that condition (5) is satisfied.

In order to show condition (6), we need to prove that if $F_{\max}(W)$ is close to $1$ then $I(W)$ is close to 0. Let $d$ be such that $F_d(W)=F_{\max}(W)$. We have: $$1-F_d(W)=\frac{1}{q}\sum_{x\in G}\Big(1-F(\rho_x,\rho_{x+d})\Big).$$ Therefore, for every $x\in G$ we have $1-F(\rho_x,\rho_{x+d})\leq q(1-F_d(W))$ and so
$$F(\rho_x,\rho_{x+d})\geq 1-q(1-F_d(W)).$$

Assume that $F_d(W)$ is high enough so that 
\begin{equation}
\label{eqAssump}
1-q(1-F_d(W))\geq \cos\frac{\pi}{2(q-1)}.
\end{equation}

Now let $x,x'\in G$ be such that $x\neq x'$. Define $A(\rho_x,\rho_{x'})=\arccos F(\rho_x,\rho_{x'})$ and let $\displaystyle l=\frac{x'-x}{d}\bmod q$. We have:
\begin{align*}
F(\rho_x,\rho_{x'})&=\cos \big(A(\rho_x,\rho_{x+ld})\big) \stackrel{(a)}{\geq} \cos\left(\sum_{i=0}^{l-1} A(\rho_{x+id},\rho_{x+(i+1)d}) \right) \\
&= \cos\left(\sum_{i=0}^{l-1} \arccos F(\rho_{x+id},\rho_{x+(i+1)d}) \right)\stackrel{(b)}{\geq} \cos\Big(l\cdot \arccos \Big(1-q\big(1-F_d(W)\big)\Big) \Big) \\
&\stackrel{(c)}{\geq} \cos\Big((q-1)\cdot \arccos \Big(1-q\big(1-F_d(W)\big)\Big) \Big),
\end{align*}
where (a) follows from the fact that $A(\rho_x,\rho_{x'})$ is a metric distance \cite{Nielsen}. (a), (b) and (c) are true because $\cos$ is a decreasing function on $\displaystyle \left[0,\frac{\pi}{2}\right]$ and we assumed Equation \eqref{eqAssump}. We deduce that
\begin{equation}
F(W)=\frac{1}{q(q-1)}\sum_{\substack{x,x'\in G,\\x\neq x'}} F(\rho_x,\rho_{x'})\geq \cos\Big((q-1)\cdot \arccos \Big(1-q\big(1-F_d(W)\big)\Big) \Big).
\label{eqFFmax}
\end{equation}
By combining Equation \eqref{eqFFmax} and inequality (iii) of Proposition \ref{propFid}, we get condition (6) of Lemma \ref{lemPolarIT}. Therefore, all the conditions of Lemma \ref{lemPolarIT} are satisfied. We conclude that $\{I(W_n)\}_{n\geq 0}$ converges almost surely to a random variable $I_\infty\in\{0,\log q\}$. This proves Equation \eqref{eqPolar}.

From Corollary \ref{corFidPolar} we can deduce that $F(W^-)\leq q^2F(W)$ and $F(W^+)\leq q^2F(W)^2$. Therefore, we can apply the same techniques that were used to prove \cite[Theorem 3.5]{SasogluThesis} in order to get Equation \eqref{eqPolarRate}.
\end{proof}

\vspace*{3mm}

Theorem \ref{thePolarFq} can be used to construct polar codes for any cq-channel whose input alphabet size is prime. The polar code construction, encoder and decoder are similar to the one described in \cite{WildeGuhaCQ}. The main idea is to send information only through synthetic cq-channels for which the symmetric Holevo information is close to $\log q$ and for which the average pairwise fidelity is less than $2^{-N^{\beta}}$, where $N=2^n$ is the blocklength of the polar code and $\beta<\frac{1}{2}$. We send frozen symbols that are known to the receiver through the remaining synthetic cq-channels. A quantum successive cancellation decoder that is similar to the one in \cite{WildeGuhaCQ} is applied. The probability of error can be shown to decay faster than $2^{-N^\beta}$ for any $\beta<\frac{1}{2}$. We postpone the accurate description and the study of the polar code till section VII where we construct polar codes in the more general case where $(G,+)$ is an arbitrary Abelian group.

\section{Polarization for arbitrary $(G,+)$}

In this section, $(G,+)$ is an arbitrary Abelian group. For every cq-channel $W: x\in G\longrightarrow \rho_x\in\mathcal{DM}(k)$ and for every subgroup $H$ of $G$, define the channel $W[H]: D\in G/H\longrightarrow \rho_D\in\mathcal{DM}(k)$ as follows:
$$\rho_D=\frac{1}{|D|}\sum_{x\in D}\rho_x.$$

$W[H]$ can be simulated as follows: if a coset $D\in G/H$ is chosen as input, a random variable $X$ is chosen uniformly from $D$ and then sent through the channel $W$.

It is easy to see that if $\displaystyle\rho^{XB}=\frac{1}{q}\sum_{x\in G}|x\rangle\langle x|^X\otimes \rho_x^B$, then $I(W[H])=I(X\bmod H; B)_\rho$.

The main result of this section is the following theorem.
\begin{mythe}
\label{thePolG}
Let $W: x\in G\longrightarrow \rho_x \in \mathcal{DM}(k)$ be a cq-channel. For every $\delta>0$, we have:
\begin{align*}
\lim_{n\to\infty} \frac{1}{2^n}\Big|\Big\{s\in\{-,+\}^n:\; &\exists H_s\;\text{a subgroup of}\; G,\\
&\big|I(W^s)-\log|G/H_s|\big|<\delta,\; \big|I(W^s[H_s])-\log|G/H_s|\big|<\delta\Big\}\Big|=1.
\end{align*}
\end{mythe}

Theorem \ref{thePolG} can be interpreted as follows: As the number of polarization steps becomes large, the synthetic cq-channels polarize to homomorphism channels projecting their input onto a quotient group of $G$. The inequality $\big|I(W^s[H_s])-\log|G/H_s|\big|<\delta$ means that from the output of $W^s$, one can determine with high probability the coset of $H_s$ to which the input belongs. The inequality $\big|I(W^s)-\log|G/H_s|\big|<\delta$ means that there is almost no other information about the input that can be determined from the output of $W^s$.

In order to prove Theorem \ref{thePolG} we need several definitions and lemmas. Let $\{B_n\}_{n\geq 0}$ be a sequence of independent and uniformly distributed $\{-,+\}$-valued random variables. Define the cq-channel-valued process $\{W_n\}_{n\geq 0}$ as follows:
\begin{itemize}
\item $W_0=W$.
\item $W_{n}=W_{n-1}^{B_n}$ for every $n\geq 1$.
\end{itemize}

\begin{mylem}
\label{lemSubMart}
For every subgroup $H$ of $G$, the process $\{I(W_n[H])\}_{n\geq 0}$ is a sub-martingale.
\end{mylem}
\begin{proof}
It is sufficient to show that $I(W^-[H])+I(W^+[H])\geq 2I(W[H])$. Let $U_1,U_2,X_1,X_2,B_1$ and $B_2$ be as in Remark \ref{remPolarSim}. We have:
\begin{align*}
I(W^-[H])+I(W^+[H])&=I(U_1\bmod H; B_1B_2) + I(U_2\bmod H; B_1B_2 U_1)\\
&\geq I(U_1\bmod H; B_1B_2) + I(U_2\bmod H; B_1B_2, U_1\bmod H)\\
&=I(U_1\bmod H,U_2\bmod H; B_1B_2)=I(X_1\bmod H,X_2\bmod H; B_1B_2)\\
&=I(X_1\bmod H, B_1) +I(X_2\bmod H; B_2) = 2I(W[H]).
\end{align*}
\end{proof}

Let $M\subset H$ be two subgroups of $G$. For every coset $D$ of $H$, let $D/M=\{C\in G/M:\; C\subset D\}$ be the set of cosets of $M$ which are subsets of $D$. Define the channel $W[M|D]: C\in D/M\longrightarrow \rho_C\in \mathcal{DM}(k)$ as follows:
$$\rho_C=\frac{1}{|C|}\sum_{x\in C}\rho_x.$$

$W[M|D]$ can be simulated as follows: if a coset $C\in D/M$ is chosen as input, a random variable $X$ is chosen uniformly from $C$ and then  sent through the channel $W$.

Define the following:
\begin{itemize}
\item $I_{M|H}(W)=I(W[M])-I(W[H])$.
\item $\displaystyle F_{\max}^{M|H}(W)=\max_{\substack{d\in H,\\d\notin M}}F_d(W)$.
\end{itemize}

The following lemma relates $I_{M|H}(W)$ to $\{I(W[M|D]):\; D\in G/H\} $.
\begin{mylem}
\label{lemIMH}
$\displaystyle I_{M|H}(W)=\frac{1}{|G/H|}\sum_{D\in G/H} I(W[M|D])$.
\end{mylem}
\begin{proof}
Let $\displaystyle\rho^{XB}=\frac{1}{q}\sum_{x\in G}|x\rangle\langle x|^X\otimes \rho_x^B$. We have $I(W[M])=I(X\bmod M;B)_{\rho}$ and $I(W[M])=I(X\bmod H;B)_{\rho}$. Therefore,
\begin{align*}
I_{M|H}(W)&=I(W[M])-I(W[H])=I(X\bmod M;B)_{\rho}-I(X\bmod M;B)_{\rho}\\
&=I(X\bmod M,X\bmod H;B)_{\rho}-I(X\bmod H;B)_{\rho}=I(X\bmod M;B|X\bmod H)_{\rho}\\
&=\sum_{D\in G/H} \frac{1}{|G/H|} I(X\bmod M;B|X\bmod H=D)_{\rho}\stackrel{(a)}{=} \sum_{D\in G/H} \frac{1}{|G/H|} I(W[M|D]),
\end{align*}
where (a) follows from the fact that conditioning on $X\bmod H=D$, the state of the input-output system becomes $\displaystyle\frac{1}{|D|}\sum_{x\in D}|x\rangle\langle x|^X\otimes \rho_x^B$ and so the mutual information between $X\bmod M$ and $B$ becomes exactly $I(W[M|D])$.
\end{proof}

The following lemma relates $F(W[M|D])$ to $F_{\max}^{M|H}(W)$.
\begin{mylem}
\label{lemFMDFmax}
For every $D\in G/H$, we have:
\begin{itemize}
\item[(1)] $\displaystyle F(W[M|D])\leq \frac{q\cdot|M|}{|H|}F_{\max}^{M|H}(W)$.
\item[(2)] There exists $\epsilon_q>0$ depending only on $q$ such that if $M$ is maximal in $H$ (i.e., $|H/M|$ is prime) and if $F_{\max}^{M|H}(W)\geq 1-\epsilon_q$, then
\begin{align*}
F(W[M|D])\geq \cos\left(\frac{|H|-|M|}{|M|} \arccos \left(1-\sqrt{1-\Big(1-q\big(1-F_{\max}^{M|H}(W)\big)\Big)^2} \right)\right).
\end{align*}
\end{itemize}
\end{mylem}
\begin{proof}
See Appendix \ref{appFMDFmax}.
\end{proof}

\begin{mylem}
\label{lemPolarIMH}
For every two subgroups $M\subset H$ of $G$ where $M$ is maximal in $H$ (i.e., $|H/M|$ is prime), the process $\{I_{M|H}(W_n)\}_{n\geq 0}$ converges almost surely to a random variable $I_{M|H}^{(\infty)}\in\{0,\log|H/M|\}$ and the process $\{F_{\max}^{M|H}(W_n)\}_{n\geq 0}$ converges almost surely to a random variable $F_{M|H}^{(\infty)}\in\{0,1\}$.
\end{mylem}
\begin{proof}
Let $I_n=I_{M|H}(W_n)$ and $T_n=F_{\max}^{M|H}(W_n)$. We will show that $I_n$ and $T_n$ satisfy the conditions of Lemma \ref{lemPolarIT}, where $q$ is replaced with $q'=|H/M|$. Conditions (1) and (3) are obviously satisfied. Condition (4) is also satisfied because of Proposition \ref{propFPolar}.

Since $I_{M|H}(W_n)=I(W_n[M])-I(W_n[H])$ and since $\{I(W_n[M])\}_{n\geq 0}$ and $\{I(W_n[H])\}_{n\geq 0}$ are sub-martingales by Lemma \ref{lemSubMart}, we conclude that $\{I_n\}_{n\geq 0}$ converges almost surely. Therefore, condition (2) is satisfied.

To see that condition (5) is satisfied, assume that $F_{\max}^{M|H}(W)$ is close to zero, then the first inequality of Lemma \ref{lemFMDFmax} implies that $F(W[M|D])$ is close to zero for every $D\in G/H$. The first inequality of Proposition \ref{propFid} then shows that $I(W[M|D])$ is close to $\log q'$, for every $D\in G/H$. Lemma \ref{lemIMH} now implies that $I_{M|H}(W)$ is close to $\log q'$.

To see that condition (6) is satisfied, assume that $F_{\max}^{M|H}(W)$ is close to 1, then the second inequality of Lemma \ref{lemFMDFmax} implies that $F(W[M|D])$ is close to 1 for every $D\in G/H$. The third inequality of Proposition \ref{propFid} then shows that $I(W[M|D])$ is close to zero, for every $D\in G/H$. Lemma \ref{lemIMH} now implies that $I_{M|H}(W)$ is close to zero.

We conclude that $\{I_{M|H}(W_n)\}_{n\geq 0}$ converges almost surely to a random variable taking values in $\{0,\log q'\}=\{0,\log|H/M|\}$ and $\{F_{\max}^{M|H}(W_n)\}_{n\geq 0}$ converges almost surely to a random variable taking values in $\{0,1\}$.
\end{proof}

\begin{mylem}
\label{lemFdSumRel}
Let $d_1,\ldots,d_r\in G$. If $\displaystyle F_{d_i}(W)\geq 1-\frac{1}{q}\left(1-\cos\frac{\pi}{2 r}\right)$ for all $1\leq i\leq r$, then
$$F_{d_1+\ldots+d_r}(W)\geq \cos\left(\sum_{i=1}^{r} \arccos \Big(1-q\big(1-F_{d_i}(W)\big)\Big) \right).$$
\end{mylem}
\begin{proof}
We may assume without loss of generality that $d_1\neq 0,\ldots,d_r\neq 0$ and $d:=d_1+\ldots+d_r\neq 0$. Define $d_1'=0$, and for every $2\leq i\leq r$, let $\displaystyle d_i'=\sum_{j=1}^{i-1} d_j$.

For every $1\leq i\leq r$, we have $\displaystyle 1-F_{d_i}(W)=\frac{1}{q}\sum_{x\in G}\big(1-F(\rho_x,\rho_{x+d_i})\big)$. Therefore, for every $x\in G$, we have $1-F(\rho_x,\rho_{x+d_i})\leq q\big(1-F_{d_i}(W)\big)$ and so $F(\rho_x,\rho_{x+d_i})\geq 1-q\big(1-F_{d_i}(W)\big) $. Therefore,

\begin{align*}
F(\rho_x,\rho_{x+d})&=F(\rho_{x+d_1'},\rho_{x+d_r'+d_r})=\cos A(\rho_{x+d_1'},\rho_{x+d_r'+d_r}) \stackrel{(a)}{\geq} \cos\left(\sum_{i=1}^{r} A(\rho_{x+d_i'},\rho_{x+d_i'+d_i})\right)\\
&= \cos\left(\sum_{i=1}^{r} \arccos F(\rho_{x+d_i'},\rho_{x+d_i'+d_i})\right) \stackrel{(b)}{\geq}  \cos\left(\sum_{i=1}^{r} \arccos \Big(1-q\big(1-F_{d_i}(W)\big)\Big) \right),
\end{align*}
where (a) follows from the fact that $A(\rho',\rho'')=\arccos F(\rho',\rho'')$ is a metric distance \cite{Nielsen}. (a) and (b) are true because $\cos$ is a decreasing function on $\displaystyle \left[0,\frac{\pi}{2}\right]$ and we assumed that $\displaystyle F_{d_i}(W)\geq 1-\frac{1}{q}\left(1-\cos\frac{\pi}{2 r}\right)$ for every $1\leq i\leq r$. We conclude that
\begin{align*}
F_d(W)=\frac{1}{q}\sum_{x\in G}F(\rho_x,\rho_{x+d})\geq \cos\left(\sum_{i=1}^{r} \arccos \Big(1-q\big(1-F_{d_i}(W)\big)\Big) \right).
\end{align*}
\end{proof}

\begin{mylem}
\label{lemFdFmaxRel}
Let $d\in G$ be such that $d\neq 0$ and let $H=\langle d\rangle$ be the subgroup generated by $d$. We have:
\begin{itemize}
\item If $F_d(W)\leq F_{\max}^{M|H}(W)$ for every maximal subgroup $M$ of $H$.
\item If $\displaystyle F_{\max}^{M|H}(W)\geq 1-\frac{1}{q}\left(1-\cos\frac{\pi}{2q}\right)$ for every maximal subgroup $M$ of $H$, then
$$F_d(W)\geq \cos\left(q\cdot \arccos \left(1-q\left(1- \min_{\substack{M\;\text{is a maximal}\\\text{subgroup of}\; H}} F_{\max}^{M|H}(W)\right)\right) \right).$$
\end{itemize}
\end{mylem}
\begin{proof}
Let $M$ be a maximal subgroup of $H$. Since $H=\langle d\rangle$, then we must have $d\in H$ and $d\notin M$. Therefore,
$$F_d(W)\leq \max_{\substack{d'\in H,\\d'\notin M}}F_{d'}(W)=F_{\max}^{M|H}(W).$$

Now let $M_1,\ldots,M_r$ be the maximal subgroups of $H=\langle d\rangle$. For every $1\leq i\leq r$, let $d_i\in H$ be such that $d_i\notin M_i$ and $F_{d_i}(W)=F_{\max}^{M_i|H}(W)$. It was shown in \cite{SahebiPradhan} that $d\in\langle d_1,\ldots,d_r\rangle$, which means that there are $l_1,\ldots,l_r\in\mathbb{N}$ such that $\displaystyle d= \sum_{i=1}^r l_i d_i$. Moreover, $l_1,\ldots,l_r\in\mathbb{N}$ can be chosen so that $l_1+\ldots+l_r\leq q$.

Since $\displaystyle F_{d_i}(W)\geq 1-\frac{1}{q}\left(1-\cos\frac{\pi}{2q}\right)\geq 1-\frac{1}{q}\left(1-\cos\frac{\pi}{2(l_1+\ldots+l_r)}\right)$ for all $1\leq i\leq r$, Lemma \ref{lemFdSumRel} implies that
\begin{align*}
F_d(W)=F_{l_1d_1+\ldots+l_rd_r}(W)&\geq \cos\left(\sum_{i=1}^{r} l_i\arccos \Big(1-q\big(1-F_{d_i}(W)\big)\Big) \right)\\
&\stackrel{(a)}{\geq} \cos\left((l_1+\ldots+l_r) \arccos\Big(1-q\big(1- \min_{1\leq i\leq r} F_{d_i}(W)\big)\Big) \right)\\
&\stackrel{(b)}{\geq} \cos\left(q\cdot\arccos\Big(1-q\big(1- \min_{1\leq i\leq r} F_{d_i}(W)\big)\Big) \right),
\end{align*}
where (a) and (b) are true because $\cos$ is decreasing on $\displaystyle\left[0,\frac{\pi}{2}\right]$ and because we assumed that  $\displaystyle F_{d_i}(W)\geq 1-\frac{1}{q}\left(1-\cos\frac{\pi}{2q}\right)$ for all $1\leq i\leq r$.
\end{proof}

\begin{myprop}
\label{propPolarFArb}
For every $d\in G$, the process $\{F_d(W_n)\}_{n\geq 0}$ converges almost surely to a random variable $F_d^{(\infty)}\in\{0,1\}$. Moreover, the random set $\{d\in G:\; F_d^{(\infty)}=1\}$ is almost surely a subgroup of $G$.
\end{myprop}
\begin{proof}
Let $d\in G$ be such that $d\neq 0$. Let $H=\langle d\rangle$ be the subgroup generated by $d$. Lemma \ref{lemPolarIMH} shows that for every maximal subgroup $M$ of $H$, the process $\left\{F_{\max}^{M|H}(W_n)\right\}_{n\geq 0}$ converges almost surely to a random variable taking values in $\{0,1\}$.

Take a sample of the process $\{W_n\}_{n\geq 0}$ for which $\left\{F_{\max}^{M|H}(W_n)\right\}_{n\geq 0}$ converges to either 0 or 1 for every maximal subgroup $M$ of $H$. We have:
\begin{itemize}
\item If there exists a maximal subgroup $M$ of $H$ for which $\left\{F_{\max}^{M|H}(W_n)\right\}_{n\geq 0}$ converges to 0, then the first point of Lemma \ref{lemFdFmaxRel} implies that $\{F_d(W_n)\}_{n\geq 0}$ converges to 0 as well.
\item If $\left\{F_{\max}^{M|H}(W_n)\right\}_{n\geq 0}$ converges to 1 for all maximal subgroups $M$ of $H$, then the second point of Lemma \ref{lemFdFmaxRel} implies that $\{F_d(W_n)\}_{n\geq 0}$ converges to 1 as well.
\end{itemize}
We conclude that for every $d\in G$, the process $\{F_d(W_n)\}_{n\geq 0}$ converges almost surely to a random variable $F_d^{(\infty)}\in\{0,1\}$. (Note that for $d=0$, we have $F_0(W_n)=1$ for all $n$.)

Now take a sample of the process $\{W_n\}_{n\geq 0}$ for which $\{F_d(W_n)\}_{n\geq 0}$ converges to either 0 or 1 for every $d\in G$. If $d_1,d_2\in G$ are such that $\{F_{d_1}(W_n)\}_{n\geq 0}$ and $\{F_{d_2}(W_n)\}_{n\geq 0}$ converge to 1, then Lemma \ref{lemFdSumRel} implies that $\{F_{d_1+d_2}(W_n)\}_{n\geq 0}$ converges to 1 as well. We conclude that the set $\big\{d\in G:\; \{F_d(W_n)\}_{n\geq 0}\;\text{converges to}\;1\big\}$ is a subgroup of $G$.
\end{proof}

\begin{mycor}
\label{corPolarF}
For every $\epsilon>0$, we have
\begin{align*}
\lim_{n\to\infty} \frac{1}{2^n}\Big|\Big\{s\in\{-,+\}^n:\; &\exists H_s\;\text{a subgroup of}\; G,\\
&F_d(W)>1-\epsilon\;\text{for every}\;d\in H_s,\;\text{and}\;F_d(W)<\epsilon\;\text{for every}\;d\notin H_s\Big\}\Big|=1.
\end{align*}
\end{mycor}

\begin{mylem}
\label{lemFromFToI}
For every $\delta>0$, there exists $\epsilon>0$ depending only on $\delta$ and $q$ such that for every cq-channel $W$, if there exists a subgroup $H$ of $G$ satisfying $F_d(W)>1-\epsilon$ for all $d\in G$ and $F_d(W)<\epsilon$ for all $d\notin H$, then $\big|I(W)-\log|G/H|\big|<\delta$ and $\big|I(W[H])-\log|G/H|\big|<\delta$.
\end{mylem}
\begin{proof}
If $H=G$, then $I(W[G])=0=\log|G/G|$ and so $\big|I(W[G])-\log|G/G|\big|=0<\delta$. On the other hand, since $H=G$, we have $F_d(W)>1-\epsilon$ for every $d\in G$. Therefore, $\displaystyle F(W)=\frac{1}{q-1}\sum_{\substack{d\in G,\\d\neq 0}}F_d(W)>1-\epsilon$. The third inequality of Proposition \ref{propFid} now implies $I(W)<\delta_q^{(1)}$ for some function $\epsilon\to \delta_q^{(1)}(\epsilon)$ (depending only on $\epsilon$ and $q$) which satisfies $\displaystyle\lim_{\epsilon\to 0} \delta_q^{(1)}(\epsilon)=0$.

Now assume that $H\neq G$. We have
\begin{align*}
F(W[H])=F(W[H|G])\stackrel{(a)}{\leq} \frac{q\cdot |H|}{q}F_{\max}^{H|G}(W)\leq  q \max_{\substack{d\in G,\\d\notin H}} F_d(W)\leq q\epsilon,
\end{align*}
where (a) follows from the first inequality of Lemma \ref{lemFMDFmax}. The first inequality of Proposition \ref{propFid} implies that $I(W[H])> \log|G/H| - \delta_q^{(2)}(\epsilon)$ for some function $\epsilon\to \delta_q^{(2)}(\epsilon)$ (depending only on $\epsilon$ and $q$) which satisfies $\displaystyle\lim_{\epsilon\to 0} \delta_q^{(2)}(\epsilon)=0$.

On the other hand, we have $\displaystyle F_{\max}^{\{0\}|H}(W)=\max_{\substack{d\in H,\\d\neq 0}} F_d(W)\geq 1-\epsilon$. Assume that $\epsilon<\epsilon_q$, where $\epsilon_q$ is given by Lemma \ref{lemFMDFmax}. For every $D\in G/H$, we have
\begin{align*}
F(W[\{0\}|D])\geq \cos\left((|H|-1)\cdot \arccos \left(1-\sqrt{1-\Big(1-q\big(1-F_{\max}^{\{0\}|H}(W)\big)\Big)^2} \right)\right).
\end{align*}
This means that $F(W[\{0\}|D])$ is close to 1 as well. The third inequality of Proposition \ref{propFid} now implies that $I(W[\{0\}|D])<\delta_q^{(3)}(\epsilon)$ for some function $\epsilon\to \delta_q^{(3)}(\epsilon)$ (depending only on $\epsilon$ and $q$) which satisfies $\displaystyle\lim_{\epsilon\to 0} \delta_q^{(3)}(\epsilon)=0$. We conclude that
$$I(W)-I(W[H])=I(W[\{0\}])-I(W[H])= I_{\{0\}|H}(W)\stackrel{(a)}{=}\frac{1}{|G/H|}\sum_{D\in G/H}I(W[\{0\}|D])< \delta_q^(3),$$
where (a) follows from Lemma \ref{lemIMH}. We conclude that
$$\big|I(W)-\log|G/H|\big|\leq |I(W)-I(W[H])| + \big|I(W[H])-\log|G/H|\big|<\delta_q^{(2)}(\delta)+\delta_q^{(3)}(\delta).$$

If we define $\delta_q(\epsilon)=\max\left\{\delta_q^{(1)}(\epsilon),\delta_q^{(2)}(\epsilon)+\delta_q^{(3)}(\epsilon)\right\}$, we get $\big|I(W)-\log|G/H|\big|<\delta_q(\epsilon)$ and $\big|I(W[H])-\log|G/H|\big|<\delta_q(\epsilon)$ in all cases. Moreover, $\displaystyle \lim_{\epsilon\to 0}\delta_q(\epsilon)=0$.

This concludes the proof of the lemma.
\end{proof}

\vspace*{3mm}

The proof of Theorem \ref{thePolG} now follows immediately from Corollary \ref{corPolarF} and Lemma \ref{lemFromFToI}.

\section{Rate of polarization}

In order to derive the rate of polarization (i.e., how fast do synthetic cq-channels polarize), we need the following two lemmas.

\begin{mylem}
\label{lemFWHPolar}
For every subgroup $H$ of $G$, we have:
\begin{itemize}
\item $F(W^-[H])\leq |H|q(q-|H|) F(W[H])$.
\item $F(W^+[H])\leq |H|(q-|H|)^2 F(W[H])^2$.
\end{itemize}
\end{mylem}
\begin{proof}
See Appendix \ref{appFWHPolar}
\end{proof}

\begin{mylem}
For any $0<\delta<\log 2$ and any $0<\beta<\frac{1}{2}$, we have
$$\lim_{n\to\infty}\frac{1}{2^n}\Big|\big\{s\in\{-,+\}^n:\;I(W^s[H])>\log |G/H|-\delta,\;F(W^s[H])\geq2^{-{2^{\beta n}}}\big\}\Big|=0.$$
\label{lemFidLemPol}
\end{mylem}
\begin{proof}
The lemma is trivial if $H=G$, so let us assume that $H\neq G$. Let $H_1,\ldots,H_r$ be a sequence of subgroups of $G$ satisfying:
\begin{itemize}
\item $H=H_1\subset \ldots\subset H_r=G$.
\item $H_i$ is maximal in $H_{i+1}$ for every $1\leq i<r$.
\end{itemize}

Let $\{W_n\}_{n\geq 0}$ be the process defined in the previous section. Lemma \ref{lemPolarIMH} implies that $\{I_{H_i|H_{i+1}}(W_n)\}_{n\geq 0}$ converges almost surely to a random variable $I_{H_i|H_{i+1}}^{(\infty)}\in\{0,\log |H_{i+1}/H_i|\}$. On the other hand, we have
\begin{align*}
I(W_n[H])=I(W_n[H])-I(W_n[G])=\sum_{i=1}^{r-1} \big(I(W_n[H_i])-I(W_n[H_{i+1}])\big)=\sum_{i=1}^{r-1} I_{H_i|H_{i+1}}(W_n).
\end{align*}
This shows that the process $\{I(W_n[H])\}_{n\geq 0}$ converges almost surely to a random variable $I_H^{(\infty)}$ satisfying $$I_H^{(\infty)}\in\{\log m:\; m\;\text{divides}\;|G/H|\}.$$

Due to the relations between the quantities $I(W)$ and $F(W)$ in Proposition \ref{propFid}, we can see that $\{F(W_n[H])\}_{n\geq0}$ converges to 0 whenever $\{I(W_n[H])\}_{n\geq0}$ converges to $\log|G/H|$, and there is a number $f_0>0$ such that $\displaystyle\liminf_{n\to\infty} F(W_n[H])>f_0$ whenever $\{I(W_n[H])\}_{n\geq0}$ converges to a number in $\{\log m:\; m\;\text{divides}\;|G/H|\}$ other than $\log|G/H|$. Therefore, we can say that almost surely, we have:
$$\lim_{n\to\infty} F(W_n[H])=0\;\;\text{or}\;\;\liminf_{n\to\infty} F(W_n[H])>f_0.$$

Now from Lemma \ref{lemFWHPolar}, we have $F(W_n^-[H])\leq q^3 F(W_n[H])$ and $F(W_n^+[H])\leq q^3 F(W_n[H])^2$. By applying exactly the same techniques that were used to prove \cite[Theorem 3.5]{SasogluThesis} we get: $$\displaystyle\lim_{n\to\infty}\mathbb{P}\Big(\Big\{I(W_n[H])>\log |G/H|-\delta, F(W_n[H])\geq2^{-{2^{n\beta}}}\Big\}\Big)=0.$$ 
By examining the explicit expression of this probability we get the lemma.
\end{proof}

\begin{mythe}
\label{thePolGF}
The polarization of $W_n$ is almost surely fast:
\begin{align*}
\lim_{n\to\infty} \frac{1}{2^n}\Big|\Big\{ & s\in\{-,+\}^n: \exists H_s\;\text{subgroup of}\;G,\\
& \big|I(W^s)-\log|G/H_s|\big|<\delta, \big|I(W^s[H_s])-\log|G/H_s|\big|<\delta, F(W^s[H_s])<2^{-2^{\beta n}} \Big\}\Big| = 1,
\end{align*}
for any $\displaystyle 0<\delta<\log 2$ and any $0<\beta<\frac{1}{2}$.
\end{mythe}
\begin{proof}
For every subgroup $H$ of $G$, define:
$$E_H=\Big\{s\in\{-,+\}^n:\;I(W^s[H])>\log |G/H|-\delta, F(W^s[H])\geq2^{-{2^{\beta n}}}\Big\},$$

\begin{align*}
E_1=\Big\{s\in\{-,+\}^n: \exists H_s\;\text{subgroup of}\;G,\big|I(W^s)-\log|G/H_s|\big|<\delta, \big|I(W^s[H_s])-\log|G/H_s|\big|<\delta \Big\},
\end{align*}
and
\begin{align*}
E_2=\Big\{ s\in\{-,+\}^n: &\exists H_s\;\text{subgroup of}\;G,\\
& \big|I(W^s)-\log|G/H_s|\big|<\delta, \big|I(W^s[H_s])-\log|G/H_s|\big|<\delta, F(W^s[H_s])<2^{-2^{\beta n}} \Big\}.
\end{align*}

If $\displaystyle s\in E_1/\Big(\bigcup_{H\;\text{subgroup of}\;G} E_H\Big)$ then  $s\in E_2$. Therefore, $$\displaystyle E_1/\Big(\bigcup_{H\;\text{subgroup of}\;G} E_H\Big) \subset E_2,$$ and $\displaystyle|E_2|\geq|E_1|-\sum_{H\;\text{subgroup of}\;G}|E_H|$. By Theorem \ref{thePolG} and Lemma \ref{lemFidLemPol} we have:
$$1\geq\lim_{n\to\infty}\frac{1}{2^n}|E_2|\geq\lim_{n\to\infty}\frac{1}{2^n}\Big(|E_1|-\sum_{H\;\text{subgroup of}\;G}|E_H|\Big)=1-0=1.$$
\end{proof}

\section{Polar code construction}
Let $W:x\in G\longrightarrow \rho_x\in\mathcal{DM}(k)$ be an arbitrary cq-channel.

Choose $\displaystyle 0<\delta<\log 2$ and $0<\beta<\beta'<\frac{1}{2}$, and let $n$ be an integer such that $$2\sqrt{2^n}\sqrt{ (q-1)2^n2^{-2^{\beta' n}}}\leq 2^{-2^{\beta n}}\;\;\text{and}\;\;\frac{1}{2^n}|E_n|>1-\frac{\delta}{2\log q},$$ where

\begin{align*}
E_n=\Big\{ s\in &\{-,+\}^n: \exists H_s\;\text{subgroup of}\;G,\\
& \big|I(W^s)-\log|G/H_s|\big|<\frac{\delta}{2}, \big|I(W^s[H_s])-\log|G/H_s|\big|<\frac{\delta}{2}, F(W^s[H_s])<2^{-2^{\beta' n}} \Big\}.
\end{align*}

Such an integer exists due to Theorem \ref{thePolGF}. For every $s\in\{-,+\}^n$ choose a subgroup $H_s$ of $G$ as follows:
\begin{itemize}
\item If $s\notin E_n$, define $H_s=G$. We clearly have $F(W^s[H_s])=0<2^{-2^{\beta' n}}$.
\item If $s\in E_n$, choose a subgroup $H_s$ of $G$ such that $F(W^s[H_s])<2^{-2^{\beta' n}}$, $\big|I(W^s)-\log|G/H_s|\big|<\frac{\delta}{2}$ and $\big|I(W^s[H_s])-\log|G/H_s|\big|<\frac{\delta}{2}$.
\end{itemize}
Now for every $s\in\{-,+\}^n$, let $f_s:G/H_s\longrightarrow G$ be a frozen mapping (in the sense that the receiver knows $f_s$) such that $f_s(a)\bmod H_s = a$ for all $a\in G/H_s$. We call such mapping \emph{a section mapping} of $G/H_s$. Let $\tilde{U}^s$ be a random coset chosen uniformly in $G/H_s$ and we let $U^s=f_s(\tilde{U}^s)$. Note that if the receiver can determine $U^s\bmod H_s=\tilde{U}^s$ accurately, then he can also determine $U^s$ since he knows $f_s$.

If $H_s\neq\{0\}$, we have some freedom on the choice of the section mapping $f_s$. We will analyze the performance of polar codes averaged on all possible section mappings. I.e., we assume that $f_s$ is chosen uniformly from all the possible section mappings of $H_s$. We can easily see that the induced distributions of $\big\{U^s:\;s\in\{-,+\}^n\big\}$ are independent and uniform in $G$. Note that for every $s\in\{-,+\}^n$, the receiver has to determine $\tilde{U}^s=U^s\bmod H_s$ in order to successfully determine $U^s$.

\subsection{Encoder}

We associate the set $S_n:=\{-,+\}^n$ with the strict total order $<$ defined as $(s_1,...,s_n)<(s_1',...,s_n')$ if and only if $s_i=-,s_i'=+$ for some $i\in\{1,...,n\}$ and $s_h=s_h'$ for all $i<h\leq n$.

For every $u=(u^s)_{s\in S_n}\in G^{S_n}$, every $0\leq n'\leq n$ and every $(s',s'')\in S_{n'}\times S_{n-n'}$, define $\mathcal{E}_{s'}^{s''}(u)\in G$ recursively on $0\leq n'\leq n$ as follows:
\begin{itemize}
\item $\mathcal{E}_{\o}^{s}(u)=u^s$ if $n'=0$ and $s\in S_n$.
\item $\mathcal{E}_{(s',-)}^{s''}(u)=\mathcal{E}_{s'}^{(s'',-)}(u)+\mathcal{E}_{s'}^{(s'',+)}(u)$ if $n'>0$, $s'\in S_{n'-1}$ and $s''\in S_{n-n'}$.
\item $\mathcal{E}_{(s',+)}^{s''}(u)=\mathcal{E}_{s'}^{(s'',+)}(u)$ if $n'>0$, $s'\in S_{n'-1}$ and $s''\in S_{n-n'}$.
\end{itemize}
For every $s\in S_n$, we write $\mathcal{E}_{\o}^{s}(u)$ as $\mathcal{E}^{s}(u)$ and $\mathcal{E}_s^{\o}(u)$ as $\mathcal{E}_s(u)$.

Let $\{W_s\}_{s\in S_n}$ be a set of $2^n$ independent copies of the channel $W$. $W_s$ should not be confused with $W^s$: $W_s$ is a copy of the channel $W$ and $W^s$ is a synthetic cq-channel obtained from $W$ as before.

Let $(U^s)_{s\in S_n}=(f_s(\tilde{U}^s))_{s\in S_n}$ be the sequence of $2^n$ independent random variables that were defined before. For every $0\leq n'\leq n$, $s'\in S_{n'}$ and $s''\in S_{n-n'}$, define $U_{s'}^{s''}=\mathcal{E}_{s'}^{s''}\big((U^s)_{s\in S_n}\big)$. We have:
\begin{itemize}
  \item $U_{\o}^{s}=U^{s}$ if $n'=0$ and $s\in\{-,+\}^n$.
  \item $U_{(s';-)}^{s''}=U_{s'}^{(s'';+)}+U_{s'}^{(s'';-)}$ if $n'>0$, $s'\in\{-,+\}^{n'-1}$ and $s''\in\{-,+\}^{n-n'}$.
  \item $U_{(s';+)}^{s''}=U_{s'}^{(s'';+)}$ if $n'>0$, $s'\in\{-,+\}^{n'-1}$ and $s''\in\{-,+\}^{n-n'}$.
\end{itemize}
For every $s\in S_n$, let $U_s=U_s^{\o}$. It is easy to see that $(U_s)_{s\in S_n}$ are independent and uniformly distributed in $G$.

For every $s\in S_n$, we send $U_s$ through the channel $W_s$. Let $B_s$ be the system describing the output of the channel $W_s$, and let $B=\{B_s\}_{s\in S_n}$. We can prove by backward induction on $n'$ that the channel $U_{s'}^{s''}\rightarrow \big(\{B_s\}_{s\;has\;s'\; as\;prefix},\{U_{s'}^{r}\}_{r<s''}\big)$ is equivalent to the channel $W^{s''}$ for every $0\leq n'\leq n$, $s'\in S_{n'}$ and $s''\in S_{n-n'}$. In particular, the channel $U^s\rightarrow \big(B,\{U^{r}\}_{r<s}\big)$ is equivalent to the channel $W^s$ for every $s\in S_n$.

Note that the encoding algorithm described above has a complexity of $O(N\log N)$, where $N=2^n$ is the blocklength of the polar code.

\subsection{Quantum successive cancellation decoder}
Before describing the decoder, let us fix a few useful notations. 

For every $s\in S_n$, define $\mathcal{L}_{s}=\{r\in S_n:\;r<s\}$ and $\mathcal{U}_{s}=\{r\in S_n:\;r>s\}$. For every $u=(u^s)_{s\in S_n}\in G^{S_n}$, define the following:
\begin{itemize}
\item For every $S\subset S_n$, let $u^S:=(u^s)_{s\in S}$.
\item For every $s\in S_n$, let $u_s:=\mathcal{E}_s(u)$.
\item Define $\displaystyle\rho_u^B:=\bigotimes_{s\in S_n}\rho_{u_s}^{B_s}$. This means that if $U^s=u^s$ for every $s\in S_n$, then the receiver sees the state $\rho_u^B$ at the output.
\end{itemize}

It is easy to see that for every $s\in S_n$, we have $W^s:u^s\in G\longrightarrow \rho^{B,U^{\mathcal{L}_s}}_{s,u^s}\in\mathcal{DM}\left(k^{2^n}\cdot q^{|\mathcal{L}_s|}\right)$, where
$$\rho^{B,U^{\mathcal{L}_s}}_{s,u^s}=\frac{1}{q^{|\mathcal{L}_s|}} \sum_{u^{\mathcal{L}_s}\in G^{{\mathcal{L}_s}}} \overline{\rho}^{B}_{u^s,u^{\mathcal{L}_s}} \otimes \left|u^{\mathcal{L}_s}\right\rangle\left\langle u^{\mathcal{L}_s}\right|^{U^{\mathcal{L}_s}},$$
and
$$\overline{\rho}^{B}_{u^s,u^{\mathcal{L}_s}}=\frac{1}{q^{|\mathcal{U}_s|}}\sum_{u^{\mathcal{U}_s}\in G^{\mathcal{U}_s}} \rho_u^B.$$

Moreover, we have $W^s[H_s]:\tilde{u}^s\in G/H_s\longrightarrow \rho^{B,U^{\mathcal{L}_s}}_{s,\tilde{u}^s}\in\mathcal{DM}\left(k^{2^n}\cdot q^{|\mathcal{L}_s|}\right)$, where
$$\rho^{B,U^{\mathcal{L}_s}}_{s,\tilde{u}^s}=\frac{1}{|H_s|}\sum_{u_s\in\tilde{u}_s}\rho^{B,U^{\mathcal{L}_s}}_{s,u^s}=\frac{1}{q^{|\mathcal{L}_s|}} \sum_{u^{\mathcal{L}_s}\in G^{{\mathcal{L}_s}}} \overline{\rho}^{B}_{\tilde{u}^s,u^{\mathcal{L}_s}} \otimes \left|u^{\mathcal{L}_s}\right\rangle\left\langle u^{\mathcal{L}_s}\right|^{U^{\mathcal{L}_s}},$$
and
$$\overline{\rho}^{B}_{\tilde{u}^s,u^{\mathcal{L}_s}}=\frac{1}{|H_s|\cdot q^{|\mathcal{U}_s|}} \sum_{u^s\in\tilde{u}^s}\sum_{u^{\mathcal{U}_s}\in G^{\mathcal{U}_s}} \rho_u^B.$$

\begin{mylem}
\label{lemPOVMFWH}
For every $u^{\mathcal{L}_s}\in G^{{\mathcal{L}_s}}$, there exists a POVM $\left\{\Pi_{(s),u^{\mathcal{L}_s},\tilde{u}^s}^{B}:\; \tilde{u}^s\in G/H_s\right\}$ such that the POVM $\left\{\Pi_{(s),\tilde{u}^s}^{B, U^{\mathcal{L}_s}}:\; \tilde{u}^s\in G/H_s\right\}$ defined as
$$\Pi_{(s),\tilde{u}^s}^{B, U^{\mathcal{L}_s}}=\sum_{u^{\mathcal{L}_s}\in G^{{\mathcal{L}_s}}}\Pi_{(s),u^{\mathcal{L}_s},\tilde{u}^s}^{B}\otimes \left|u^{\mathcal{L}_s}\right\rangle\left\langle u^{\mathcal{L}_s}\right|^{U^{\mathcal{L}_s}}, $$
satisfies
$$1-\frac{1}{|G/H_s|}\sum_{\tilde{u}^s\in G/H_s} \Tr\left(\Pi_{(s),\tilde{u}^s}^{B, U^{\mathcal{L}_s}} \rho^{B,U^{\mathcal{L}_s}}_{s,\tilde{u}^s}\right)<(|G/H_s|-1)F(W[H_s]).$$
\end{mylem}
\begin{proof}
See Appendix \ref{appPOVMFWH}.
\end{proof}

\vspace*{3mm}

For every $s\in S_n$, every $u^s\in G$ and every $u^{\mathcal{L}_s}\in G^{\mathcal{L}_s}$, define the POVM $\left\{\Pi_{(s),u^{\mathcal{L}_s},u^s}^{B}:\; u^s\in G\right\}$ as:
$$\Pi_{(s),u^{\mathcal{L}_s},u^s}^{B}=\begin{cases}\Pi_{(s),u^{\mathcal{L}_s},u^s\bmod H_s}^{B}\quad&\text{if}\;u^s=f_s(u^s\bmod H_s),\\0\quad&\text{otherwise}. \end{cases}$$

Now we are ready to describe the quantum successive cancellation decoder. We will decode $\{U^s\}_{s\in S_n}$ successively by respecting the order $<$ on $S_n$. At the stage $s\in S_n$, we would have decoded $U^{\mathcal{L}_s}=(U^{r})_{r<s}$ and obtained an estimate $\hat{u}^{\mathcal{L}_s}=(\hat{u}^{r})_{r<s}$ of it, so we apply the POVM $\left\{\Pi_{(s),\hat{u}^{\mathcal{L}_s},u^s}^{B}:\; u^s\in G\right\}$ on the output system $B=(B_s)_{s\in S_n}$ and we let $\hat{u}^s$ be the measurement result. We assume that the the POVM measurement is designed so that if $\sigma^B$ was the state of the $B$ system before the measurement, and if the output $\hat{u}^s$ occurs, then the post-measurement state is $\displaystyle \frac{\sqrt{\Pi_{(s),\hat{u}^{\mathcal{L}_s},\hat{u}^s}^{B}}\sigma^B\sqrt{\Pi_{(s),\hat{u}^{\mathcal{L}_s},\hat{u}^s}^{B}}}{\Tr\left(\Pi_{(s),\hat{u}^{\mathcal{L}_s},\hat{u}^s}^{B}\sigma^B\right)}$.

The whole procedure is equivalent to applying the POVM $\left\{\Lambda_{u}^B:\;u=(u^s)_{s\in S_n}\in G^{S_n}\right\}$ defined as:
\begin{align*}
\Lambda_{u}^B=\sqrt{\Pi^B_{(s_1),u^{s_1}}}\ldots\sqrt{\Pi^B_{(s_i),u^{\mathcal{L}_{s_i}},u^{s_i}}}\ldots \sqrt{\Pi^B_{(s_N),u^{\mathcal{L}_{s_N}},u^{s_N}}}&\sqrt{\Pi^B_{(s_N),u^{\mathcal{L}_{s_N}},u^{s_N}}} \ldots\\
&\sqrt{\Pi^B_{(s_i),u^{\mathcal{L}_{s_i}},u^{s_i}}}\ldots \sqrt{\Pi^B_{(s_1),u^{s_1}}},
\end{align*}

where $s_1< s_2<\ldots< s_N$ are the $N=2^n$ elements of $S_n$ ordered according to the order relation $<$.

It is easy to see that $\Lambda_{u}\geq 0$ for every $u\in G^{S_n}$, and $\displaystyle\sum_{u\in G^{S_n}} \Lambda_u=I$.

\subsection{Performance of polar codes}

For every $s\in S_n$, let $\mathcal{F}_s$ be the set of section mappings of $G/H_s$. We have:
$$\mathcal{F}_s=\left\{f_s\in G^{G/H_s}:\; f_s(\hat{u}^s)\in \hat{u}^s\;\text{for all}\; \hat{u}^s\in G/H_s\right\}.$$
It is easy to see that $|\mathcal{F}_s|=|H_s|^{|G/H_s|}$. Define
$$\mathcal{F}:=\prod_{s\in S_n}\mathcal{F}_s.$$

For every $f=(f_s)_{s\in S_n}\in \mathcal{F}$ and every $\displaystyle\tilde{u}=(\tilde{u}^s)_{s\in S_n}\in\prod_{s\in S_n} (G/H_s)$, define $f(\tilde{u})=\big(f_s(\tilde{u}^s)\big)_{s\in S_n}\in G^{S_n}$.

The probability of error of the quantum successive cancellation decoder for a particular choice of $\displaystyle f=(f_s)_{s\in S_n}\in\mathcal{F}=\prod_{s\in S_n}\mathcal{F}_s$ is given by:

\begin{align*}
P_e(f)&=\frac{1}{\prod_{s\in S_n} |G/H_s|}\sum_{\tilde{u}\in \prod_{s\in S_n}(G/H_s)}\left(1-\Tr\left(\Lambda_{f(\tilde{u})}^B\rho_{f(\tilde{u})}^B\right)\right)=\mathbb{E}_{\tilde{U}}\left(1-\Tr\left(\Lambda_{f(\tilde{U})}^B\rho_{f(\tilde{U})}^B\right)\right),
\end{align*}
where $\tilde{U}=(\tilde{U}^s)_{s\in S_n}$ is uniformly distributed in $\displaystyle \prod_{s\in S_n}(G/H_s)$.

The probability of error averaged over all the choices of $\displaystyle f=(f_s)_{s\in S_n}\in\mathcal{F}=\prod_{s\in S_n}\mathcal{F}_s$ is:

\begin{align*}
\overline{P}_e&=\frac{1}{|\mathcal{F}|}\sum_{f\in \mathcal{F}}P_e(f)=\frac{1}{|\mathcal{F}|}\sum_{f\in \mathcal{F}} \mathbb{E}_{\tilde{U}}\left(1-\Tr\left(\Lambda_{f(\tilde{U})}\rho_{f(\tilde{U})}^B\right)\right)\\
&=\mathbb{E}_{F,\tilde{U}}\left(1-\Tr\left(\Lambda_{F(\tilde{U})}^B\rho_{F(\tilde{U})}^B\right)\right)=\mathbb{E}_{F,\tilde{U}}\left(1-\Tr\left(\Lambda_{U}^B\rho_{U}^B\right)\right),
\end{align*}
where $F=(F_s)_{s\in S_n}$ is uniformly distributed in $\displaystyle\mathcal{F}=\prod_{s\in S_n}\mathcal{F}_s$, and $U=(U^s)_{s\in S_n}=F(U)=\big(F_s(\tilde{U}^s)\big)_{s\in S_n}$.  It is easy to see that $\{U^s:\;s\in S_n\}$ are independent and uniformly distributed in $G$. We have:

\begin{align*}
\overline{P}_e&=\mathbb{E}_{F,\tilde{U}}\left(1-\Tr\left(\Lambda_{U}^B\rho_{U}^B\right)\right)\\
&=\mathbb{E}_{F,\tilde{U}}\left(1- \Tr\left(\sqrt{\Pi^B_{(s_N),U^{\mathcal{L}_{s_N}},U^{s_N}}}\ldots\sqrt{\Pi^B_{(s_1),U^{s_1}}} \rho_U^B\sqrt{\Pi^B_{(s_1),U^{s_1}}}
\ldots \sqrt{\Pi^B_{(s_N),U^{\mathcal{L}_{s_N}},U^{s_N}}}\right)\right)\\
&\stackrel{(a)}{\leq} \mathbb{E}_{F,\tilde{U}}\left( 2\sqrt{N} \sqrt{\sum_{i=1}^N \left(1-\Tr\left(\Pi^B_{(s_i),U^{\mathcal{L}_{s_i}},U^{s_i}}\rho_U^B\right)\right) }\right)\\
&\stackrel{(b)}{\leq} 2\sqrt{N} \sqrt{\mathbb{E}_{F,\tilde{U}}\left( \sum_{i=1}^N \left(1-\Tr\left(\Pi^B_{(s_i),U^{\mathcal{L}_{s_i}},U^{s_i}}\rho_U^B\right)\right) \right)}=2\sqrt{N} \sqrt{\sum_{s\in S_n}  \mathbb{E}_{F,\tilde{U}} \left(1-\Tr\left(\Pi^B_{(s),U^{\mathcal{L}_{s}},U^{s}}\rho_U^B\right) \right)}\\
&\stackrel{(c)}{=}2\sqrt{N} \sqrt{\sum_{s\in S_n}  \mathbb{E}_{F,\tilde{U}} \left(1-\Tr\left(\Pi^B_{(s),U^{\mathcal{L}_{s}},\tilde{U}^{s}}\rho_U^B\right) \right)}\stackrel{(d)}{=}2\sqrt{N} \sqrt{\sum_{s\in S_n}  \mathbb{E}_{U,\tilde{U}^s} \left(1-\Tr\left(\Pi^B_{(s),U^{\mathcal{L}_{s}},\tilde{U}^{s}}\rho_U^B\right) \right)}\\
&=2\sqrt{N} \sqrt{\sum_{s\in S_n} \bigg(1-\mathbb{E}_{\tilde{U}^s,U^{\mathcal{L}_s}}\Tr\Big(\Pi^B_{(s),U^{\mathcal{L}_{s}},\tilde{U}^{s}} \mathbb{E}_{U^s,U^{\mathcal{U}_s}|\tilde{U}^s,U^{\mathcal{L}_s}} \left(\rho_U^B\right)\Big)\bigg) }\\
&\stackrel{(e)}{=}2\sqrt{N} \sqrt{\sum_{s\in S_n} \bigg(1-\mathbb{E}_{\tilde{U}^s,U^{\mathcal{L}_s}}\left(\Tr\left(\Pi^B_{(s),U^{\mathcal{L}_{s}},\tilde{U}^{s}} \overline{\rho}^B_{\tilde{U}^s,U^{\mathcal{L}_s}}\right)\right) \bigg)},
\end{align*}
where (a) follows from the ``non-commutative union bound" of Lemma \ref{lemNonComUnionBound}. (b) follows from the concavity of the square root. (c) follows from the fact that $U^s=f_s(\tilde{U}^s)$, which implies that $U^s\bmod H_s= \tilde{U}^s$ and $U^s=f_s(U^s\bmod H_s)$, which in turn implies that $\Pi^B_{(s),U^{\mathcal{L}_{s}},U^s}=\Pi^B_{(s),U^{\mathcal{L}_{s}},U^s\bmod H_s}$. (d) follows from the fact that $\Tr\left(\Pi^B_{(s),U^{\mathcal{L}_{s}},\tilde{U}^{s}}\rho_U^B\right)$ depends only on $\tilde{U}^s$ and $U$. (e) follows from the fact that for every $\tilde{u}^s\in G/H_s$ and every $u^{\mathcal{L}_s}\in G^{\mathcal{L}_s}$, we have:

\begin{align*}
\mathbb{E}_{U^s,U^{\mathcal{U}_s}|\tilde{U}^s=\tilde{u}^s,U^{\mathcal{L}_s}=u^{\mathcal{L}_s}} \left(\rho_U^B\right)&=\frac{1}{|H_s|\cdot q^{|\mathcal{U}_s|}}\sum_{u^s\in\tilde{u}^s}\sum_{u^{\mathcal{U}_s}\in G^{\mathcal{U}_s}} \rho^B_u=\overline{\rho}^B_{\tilde{u}^s,u^{\mathcal{L}_s}}.
\end{align*}

On the other hand, we have:
\begin{align*}
\mathbb{E}_{\tilde{U}^s,U^{\mathcal{L}_s}}\left(\Tr\left(\Pi^B_{(s),U^{\mathcal{L}_{s}},\tilde{U}^{s}} \overline{\rho}^B_{\tilde{U}^s,U^{\mathcal{L}_s}}\right)\right) &= \frac{1}{|G/H_s|}\sum_{\tilde{u}^s\in G/H_s}\frac{1}{q^{|\mathcal{L}_s|}}\sum_{u^{\mathcal{L}_s}\in G^{\mathcal{L}_s}} \Tr\left( \Pi^B_{(s),u^{\mathcal{L}_{s}},\tilde{u}^{s}} \overline{\rho}^B_{\tilde{u}^s,u^{\mathcal{L}_s}}\right)\\
&=\frac{1}{|G/H_s|}\sum_{\tilde{u}^s\in G/H_s} \Tr\left(\Pi^{B,U^{\mathcal{L}_s}}_{(s),\tilde{u}^{s}}  \rho^{B,U^{\mathcal{L}_s}}_{s,\tilde{u}^s}\right).
\end{align*}

Therefore,

\begin{align*}
\overline{P}_e&\leq 2\sqrt{N} \sqrt{\sum_{s\in S_n} \bigg(1-\mathbb{E}_{\tilde{U}^s,U^{\mathcal{L}_s}}\left(\Tr\left(\Pi^B_{(s),U^{\mathcal{L}_{s}},\tilde{U}^{s}} \overline{\rho}^B_{\tilde{U}^s,U^{\mathcal{L}_s}}\right)\right) \bigg)}\\
&=2\sqrt{N} \sqrt{\sum_{s\in S_n} \left(1-\frac{1}{|G/H_s|}\sum_{\tilde{u}^s\in G/H_s} \Tr\left(\Pi^{B,U^{\mathcal{L}_s}}_{(s),\tilde{u}^{s}}  \rho^{B,U^{\mathcal{L}_s}}_{s,\tilde{u}^s}\right) \right)}\\
&\stackrel{(a)}{\leq} 2\sqrt{N}\sqrt{\sum_{s\in S_n} (|G/H_s|-1)F(W[H_s])}\leq 2\sqrt{N}\sqrt{\sum_{s\in S_n} (q-1)2^{-2^{\beta' n}}}\\
&\leq 2\sqrt{2^n}\sqrt{ (q-1)2^n2^{-2^{\beta' n}}}\leq 2^{-2^{\beta n}},
\end{align*}
where (a) follows from Lemma \ref{lemPOVMFWH}.

The above upper bound was calculated on average over a random choice of the frozen section mappings. Therefore, there is at least one choice of the frozen section mappings for which the upper bound of the probability of error still holds.

It remains to study the rate of the constructed polar code. The rate at which we are communicating is $\displaystyle R=\frac{1}{2^n}\sum_{s\in \{-,+\}^n}\log |G/H_s|=\frac{1}{2^n}\sum_{s\in E_n}\log |G/H_s|$. On the other hand, we have $\big|I(W^s)-\log|G/H_s|\big|<\frac{\delta}{2}$ for all $s\in E_n$. Now since we have $\displaystyle\sum_{s\in\{-,+\}^n} I(W^s)=2^nI(W)$, we conclude that:
\begin{align*}
I(W)&=\frac{1}{2^n}\sum_{s\in \{-,+\}^n} I(W^s)= \frac{1}{2^n}\sum_{s\in E_n}I(W^s) + \frac{1}{2^n}\sum_{s\in E_n^c}I(W^s)\\
&<\frac{1}{2^n}\sum_{s\in E_n}\Big(\log|G/H_s|+\frac{\delta}{2}\Big) + \frac{1}{2^n}|E_n^c|\log q \\
&< R + \frac{1}{2^n}|E_n|\frac{\delta}{2} + \frac{\delta}{2\log q}\log q\\
&\leq R+\frac{\delta}{2}+\frac{\delta}{2}=R+\delta,
\end{align*}
where $E_n^c=\{-,+\}^n\setminus E_n$.

To this end we have proven the following theorem which is the main result of this paper:

\begin{mythe}
Let $W:x\in G\longrightarrow \rho_x\in\mathcal{DM}(k)$ be an arbitrary cq-channel, where the input alphabet is endowed with an Abelian group operation. For every $\delta>0$ and every $0<\beta<\frac{1}{2}$, there exists a polar code of blocklength $N=2^n$ based on the group operation which has a rate $R>I(W)-\delta$ and an encoder algorithm of complexity $O(N\log N)$. Moreover, the probability of error of the quantum successive cancellation decoder is less than $2^{-N^\beta}$.
\end{mythe}

\section{Polar codes for arbitrary classical-quantum MACs}

An $m$-user classical-quantum multiples access channel (cq-MAC) $$W:(x_1,\ldots,x_m)\in G_1\times\ldots\times G_m\longrightarrow \rho_{x_1,\ldots,x_m}\in\mathcal{DM}(k)$$ takes classical inputs $\{x_i\in G_i:\;1\leq i\leq m\}$ from the $m$ users and produces a quantum output $\rho_{x_1,\ldots,x_m} \in \mathcal{DM}(k)$. We assume that the input alphabets $G_i$ are finite but their sizes $q_i=|G_i|$ can be arbitrary.

The achievable rate-region is described by a collection of inequalities \cite{cqMAC}:

$$\forall S\subset\{1,\ldots,m\},\; 0\leq R_S\leq I(X_S;B|X_{S^c})_{\rho}=I(X_S;BX_{S^c})_{\rho},$$
where $\displaystyle R_S=\sum_{i\in S}R_i$, $X_S=(X_i)_{i\in S}$, $S^c=\{1,\ldots,m\}\setminus S$, and the mutual information $I(X_S;Y|X_{S^c})_{\rho}$ is computed according to following state:
$$\rho^{X_1,\ldots,X_m,B}=\sum_{\substack{x_1\in G_1,\\\vdots\\x_m\in G_m}}\left(\prod_{i=1}^m\mathbb{P}_{X_i}(x_i)\right)\left(\bigotimes_{1\leq i\leq m}|x_i\rangle\langle x_i|^{X_i}\right)\otimes \rho^B_{x_1,\ldots,x_m},$$
for some independent probability distributions $\{\mathbb{P}_{X_i}(x_i):\;x_i\in G_i\}$ on $G_i$ for $1\leq i\leq m$.

We are interested in the case where the probability distributions of $X_1,\ldots,X_m$ are uniform in $G_1,\ldots,G_m$ respectively. We define the symmetric capacity region $\mathcal{J}(W)$ of $W$ as
$$\mathcal{J}(W)=\Big\{(R_1,\ldots,R_m)\in\mathbb{R}^m:\; 0\leq R_S\leq I[S](W),\;\forall S\subset\{1,\ldots,m\}\Big\},$$
where $I[S](W):=I(X_S;BX_{S^c})_{\rho}$ is computed according to
$$\rho^{X_1,\ldots,X_m,B}=\frac{1}{q_1\cdots q_m}\sum_{\substack{x_1\in G_1,\\\vdots\\x_m\in G_m}}\left(\bigotimes_{1\leq i\leq m}|x_i\rangle\langle x_i|^{X_i}\right)\otimes \rho^B_{x_1,\ldots,x_m}.$$
The set $\big\{(R_1,\ldots,R_m)\in \mathcal{J}(W):\; R_1+\ldots+R_m=I(W)\big\}$ is called the dominant face of $\mathcal{J}(W)$, where $I(W):=I[\{1,\ldots,m\}](W)=I(X_1\ldots X_m;B)_{\rho}$ is the symmetric sum-capacity of $W$.

For every $1\leq i\leq m$, we fix an Abelian group operation on $G_i$ and we denote it additively. It is possible to construct cq-MAC codes which achieve the rates in the region $\mathcal{J}(W)$ using one of the following two methods:
\begin{itemize}
\item By using the monotone chain rule method of Ar{\i}kan \cite{Monotone} and applying a polarization transformation using the Abelian group operation for each user.
\item By using the rate-splitting method described in \cite{SasogluTelYeh} and applying a polarization transformation using the Abelian group operation for each user.
\end{itemize}
By using the cq-channel polarization results of this paper and a similar analysis as in \cite{Monotone}, \cite{SasogluTelYeh} and \cite{PolarcqMAC}, we can show that both methods yield cq-MAC codes that achieve the whole region $\mathcal{J}(W)$ for which the probability of error of the quantum successive cancellation decoder decays faster than $2^{-N^\beta}$ for any $\beta<\frac{1}{2}$, where $N$ is the blocklength of the code.

However, one may hesitate to call the codes obtained using these methods as cq-MAC polar codes because they are not based on the polarization of cq-MACs. These methods are hybrid schemes which combine cq-channel polarization (not cq-MAC polarization) with other techniques. Moreover, the code construction for these methods is more complicated than cq-MAC polar codes. In the rest of this section, we describe how cq-MAC polar codes are constructed.

We define the cq-MACs $W^-$ and $W^+$ as follows:
$$W^-:(u_{1,1},\ldots,u_{1,m})\in G_1\times\ldots\times G_m\longrightarrow \rho_{u_{1,1},\ldots,u_{1,m}}^-\in\mathcal{DM}(k^2),$$
$$W^+:(u_{2,1},\ldots,u_{2,m})\in G_1\times\ldots\times G_m\longrightarrow \rho_{u_{2,1},\ldots,u_{2,m}}^+\in\mathcal{DM}(k^2q_1\cdots q_m),$$
where
$$\rho_{u_{1,1},\ldots,u_{1,m}}^-=\frac{1}{q_1\cdots q_m} \sum_{\substack{u_{2,1}\in G_1,\\\vdots\\u_{2,m}\in G_m}} \rho_{u_{1,1}+u_{2,1},\ldots,u_{1,m}+u_{2,m}}\otimes \rho_{u_{2,1},\ldots,u_{2,m}},$$
and
$$\rho_{u_{2,1},\ldots,u_{2,m}}^+=\frac{1}{q_1\cdots q_m} \sum_{\substack{u_{1,1}\in G_1,\\\vdots\\u_{1,m}\in G_m}} \rho_{u_{1,1}+u_{2,1},\ldots,u_{1,m}+u_{2,m}}\otimes \rho_{u_{2,1},\ldots,u_{2,m}}\otimes\left(\bigotimes_{1\leq i\leq m}|u_{1,i}\rangle\langle u_{1,i}|\right).$$

Note that the cq-MAC $W$ can be seen as a cq-channel with input in $G:=G_1\times\ldots\times G_m$. Moreover, $W^-$ and $W^+$ when seen as cq-channels can be obtained from the cq-channel $W$ by applying the polarization transformation which uses the Abelian group operation of the product group $G$. Therefore, the cq-channel polarization results of the previous sections can be applied to $W$. In particular, we have:
\begin{itemize}
\item $I(W^-)+I(W^+)=2I(W)$. This shows that the symmetric sum-capacity is conserved by the polarization transformation and that for every $n>0$, the region $\displaystyle\frac{1}{2^n}\sum_{s\in\{-,+\}^n}\mathcal{J}(W^s)$ contains points on the dominant face of $\mathcal{J}(W)$.
\item For every subgroup $H$ of $G$, we have $I(W^-[H])+I(W^-[H])\geq 2I(W[H])$ by Lemma \ref{lemSubMart}. Therefore, for every $S\subset\{1,\ldots,m\}$, we have
\begin{equation}
\label{eqSideRate}
\begin{aligned}
I[S](W^-)+I[S](W^+)&=\left(I(W^-)-I(W^-[G_{S}])\right)+\left(I(W^+)-I(W^+[G_{S}])\right)\\
&\leq 2I(W) - 2I(W[G_{S}])=2I[S](W),
\end{aligned}
\end{equation}
where,
$$G_{S}=\left(\prod_{i\in S} G_i\right)\times\left(\prod_{j\notin S} \{0\}\right).$$
Equation \eqref{eqSideRate} shows that although the symmetric-sum capacity is conserved by polarization, the highest achievable individual rates can decrease. In other words, polarization can induce a loss in the symmetric capacity region.
\item Theorem \ref{thePolGF} implies that
\begin{align*}
\lim_{n\to\infty} \frac{1}{2^n}\Big|\big\{ & s\in\{-,+\}^n: \exists H_s\;\text{subgroup of}\;G,\\
& \big|I(W^s)-\log|G/H_s|\big|<\delta, \big|I(W^s[H_s])-\log|G/H_s|\big|<\delta, F(W^s[H_s])<2^{-2^{\beta n}} \big\}\Big| = 1.
\end{align*}
In other words, as the number of polarization steps becomes large, the synthetic cq-MACs become close to deterministic homomorphism channels which project the input $(U^s_1,\ldots,U^s_m)$ onto some quotient group $G/H_s$ of the product group $G$.
\end{itemize}

One can employ the properties of subgroups of product groups to show that the polarized cq-MAC $W^s$ is an ``easy" cq-MAC in a sense similar to the way easy MACs were defined in \cite{RajErgII}. This allows the construction of cq-MAC polar codes for which the probability of error of the quantum successive cancellation decoder decays faster than $2^{-N^{\beta}}$ for any $0<\beta<\frac{1}{2}$, where $N=2^n$ is the blocklength of the code. The region of rates that are achievable by cq-MAC polar codes is given by:
\begin{align*}
\mathcal{J}^{\text{pol}}(W)&=\bigcap_{n\geq 0} \left(\frac{1}{2^n}\sum_{s\in\{-,+\}^n} \mathcal{J}(W^s)\right)\\
&=\Big\{(R_1,\ldots,R_m)\in\mathbb{R}^m:\; R_S\leq I^{\text{pol}}[S](W),\;\forall S\subset\{1,\ldots,m\}\Big\},
\end{align*}
where
$$I^{\text{pol}}[S](W)=\lim_{n\to\infty}\frac{1}{2^n}\sum_{s\in\{-,+\}^n}I[S](W^s).$$

The cq-MAC polar codes can be compared to the two cq-MAC coding methods that were described at the beginning of this section:
\begin{itemize}
\item The cq-MAC polar codes has the advantage that the code construction is simpler.
\item The other two coding methods have the advantage that they always achieve the whole symmetric capacity region $\mathcal{J}(W)$, which may not be the case for cq-MAC polar codes in general.
\end{itemize}

\section{Conclusion}

We have shown that using a polarization transformation that is based on an Abelian group operation on the input alphabet yields multi-level polarization for arbitrary classical-quantum channels in a similar way as in the case of classical channels. This result made it possible to construct polar codes for arbitrary cq-channels and arbitrary cq-MACs.

One weakness of the results presented here is that the proposed quantum successive cancellation decoder does not seem to have an efficient implementation. This was also the case for the polar codes that were constructed for binary-input cq-channels in \cite{WildeGuhaCQ}. Finding an efficient decoder for the polar codes remains an open problem.

If we define cq-polarizing binary operations as those which can polarize an arbitrary cq-channel to ``easy" cq-channels in a sense similar to the definition of classical polarizing binary operations \cite{RajErgII}, then this paper has shown that Abelian group operations are cq-polarizing. Therefore, being an Abelian group operation is a sufficient condition to be cq-polarizing. On the other hand, from the results of \cite{RajErgII} we can deduce that being uniformity-preserving and having a right-inverse that is strongly ergodic are necessary conditions because classical channels are a particular case of cq-channels. Finding a necessary and sufficient condition for a binary operation to be cq-polarizing remains an open problem. Trying to prove a quantum version of the results in \cite{RajErgII} by using a similar approach may not be successful because the proof of the sufficient condition in \cite{RajErgII} relies heavily on the entropy of the input conditioned on a particular output symbol, and this does not have an analogue in the case of cq-channels.

We have shown that cq-MAC polarization can induce a loss in the symmetric capacity region. A necessary and sufficient condition for $\mathcal{J}^{\text{pol}}(W)=\mathcal{J}(W)$ in the case of classical MACs was given in \cite{RajTelFourier}. Generalizing the results of \cite{RajTelFourier} to cq-MACs is an open problem. We note that the condition in \cite{RajTelFourier} was given in terms of the Fourier transform of the probability distribution of one input conditioned on the output and on the other input. Since this conditional probability does not have an analogue in the case of cq-MACs, generalizing the results of \cite{RajTelFourier} to cq-MACs might be challenging and a completely different approach might be needed.

\appendices

\section{Proof of Proposition \ref{propFid}}
\label{appFid}

In \cite[Prop. 1]{HolevoRel}, it was shown that for every $0\leq s\leq 1$, we have:
$$I(W)\geq -\frac{1}{s}\log\Tr\left[\left(\sum_{x\in G}\mathbb{P}_X(x)\cdot\rho_x^{\frac{1}{1+s}}\right)^{1+s}\right].$$

By taking $s=1$, we obtain:

\begin{align*}
I(W)&\geq -\log\Tr\left[\left(\sum_{x\in G}\frac{1}{q}\cdot\sqrt{\rho_x}\right)^2\right]= -\log\Tr\left(\frac{1}{q^2}\sum_{x,x'\in G}\sqrt{\rho_x}\sqrt{\rho_{x'}}\right)\\
&=-\log\Tr\left(\frac{1}{q^2}\sum_{x\in G}\rho_x + \frac{1}{q^2}\sum_{\substack{x,x'\in G,\\x\neq x'}}\sqrt{\rho_x}\sqrt{\rho_{x'}}\right)=-\log\left(\frac{1}{q} + \frac{1}{q^2}\sum_{\substack{x,x'\in G,\\x\neq x'}}\Tr(\sqrt{\rho_x}\sqrt{\rho_{x'}})\right)\\
&\stackrel{(a)}{\geq} -\log\left(\frac{1}{q} + \frac{1}{q^2}\sum_{\substack{x,x'\in G,\\x\neq x'}}F(\rho_x,\rho_{x'})\right)= \log \frac{q}{1+(q-1)F(W)},
\end{align*}
where (a) follows from the fact that $\Tr(\sqrt{\rho_x}\sqrt{\rho_{x'}})\leq \Tr(|\sqrt{\rho_x}\sqrt{\rho_{x'}}|)=\|\sqrt{\rho_x}\sqrt{\rho_{x'}}\|_1=F(\rho_x,\rho_{x'})$.

\vspace*{3mm}

In order to prove the second inequality, define the channel $\tilde{W}:x\in G\longrightarrow \tilde{\rho}_x\in\mathcal{DM}(k\cdot q^2)$ as follows:

$$\tilde{\rho}_x^{BS_1S_2}=\rho_x^B\otimes \left(\frac{1}{2(q-1)}\sum_{\substack{x'\in G,\\ x'\neq x}}\Big(|x\rangle\langle x|^{S_1}\otimes|x'\rangle\langle x'|^{S_2}+|x'\rangle\langle x'|^{S_1}\otimes|x\rangle\langle x|^{S_2}\Big)\right).$$

The two additional systems $S_1$ and $S_2$ can be interpreted as additional side information about the input which is provided to the receiver. Note that if $S_1S_2$ are traced out, we recover the channel $W$.

Let $\displaystyle \tilde{\rho}^{XBS_1S_2}=\frac{1}{q}\sum_{x\in G}|x\rangle\langle x|^X\otimes \tilde{\rho}_x^{BS_1S_2}$. We have:

\begin{align*}
I(W) &= I(X;B)_{\tilde{\rho}}\leq I(X;BS_1S_2)_{\tilde{\rho}} =  I(X;S_1S_2)_{\tilde{\rho}} + I(X;B|S_1S_2)_{\tilde{\rho}} \\
&= H(X)- H(X|S_1S_2) + I(X;B|S_1S_2)_{\tilde{\rho}}\\
&\stackrel{(a)}{=}\log(q)-\log(2) + \sum_{s_1,s_2\in G} I(X;B|S_1=s_1,S_2=s_2)\mathbb{P}_{S_1,S_2}(s_1,s_2)\\
&\stackrel{(b)}{=}\log(q/2) + \frac{1}{q(q-1)} \sum_{\substack{s_1,s_2\in G,\\s_1\neq s_2}} I(X;B|S_1=s_1,S_2=s_2)\\
&\stackrel{(c)}{=}\log(q/2) + \frac{1}{q(q-1)} \sum_{\substack{s_1,s_2\in G,\\s_1\neq s_2}} I(W_{s_1,s_2}),
\end{align*}
where (a) follows from the fact that given $\{S_1=s_1,S_2=s_2\}$, the conditional probability distribution of $X$ is uniform in $\{s_1,s_2\}$. (b) follows from the fact that the distribution of $(S_1,S_2)$ is uniform in the set $$\{(s_1,s_2)\in G\times G:\; s_1\neq s_2\}.$$  (c) is true because conditioning $\tilde{\rho}^{XBS_1S_2}$ on $\{S_1=s_2,S_2=s_2\}$ and then tracing out $S_1S_2$ gives the state $\displaystyle \frac{1}{2}|s_1\rangle\langle s_1|^X\otimes \rho_{s_1}^B+\frac{1}{2}|s_2\rangle\langle s_2|^X\otimes \rho_{s_2}^B$ which just represents $W_{s_1,s_2}$ with uniform input, where $W_{s_1,s_2}:x\in\{0,1\}\longrightarrow \rho_{x,s_1,s_2}\in\mathcal{DM}(k)$ is the binary-input cq-channel defined as $\rho_{0,s_1,s_2}=\rho_{s_1}$ and $\rho_{1,s_1,s_2}=\rho_{s_2}$. In other words, the channel $W_{s_1,s_2}$ is obtained from $W$ by restricting the input to $\{s_1,s_2\}$.

Now since $W_{s_1,s_2}$ is a binary-input cq-channel, we have from \cite[Prop. 1]{WildeGuhaCQ} that $$I(W_{s_1,s_2})\leq (\log 2)\sqrt{1-F(W_{s_1,s_2})^2}=(\log 2)\sqrt{1-F(\rho_{s_1},\rho_{s_2})^2}.$$ Therefore,
\begin{align*}
I(W) \leq \log(q/2) + \frac{1}{q(q-1)} \sum_{\substack{s_1,s_2\in G,\\s_1\neq s_2}} (\log 2)\sqrt{1-F(\rho_{s_1},\rho_{s_2})^2}\leq \log(q/2) + (\log 2)\sqrt{1-F(W)^2},
\end{align*}
where the last inequality follows from the concavity of the function $t\to \sqrt{1-t^2}$.

\vspace*{3mm}

It remains to show the last inequality of Proposition \ref{propFid}. Define the following:
\begin{itemize}
\item $\displaystyle\rho^{XB}=\frac{1}{q}\sum_{x\in G}|x\rangle\langle x|^X\otimes \rho_x^B$.
\item $\displaystyle\Lambda^{XB} = \sum_{x\in G} |x\rangle\langle x|^X\otimes E_x^B $, where $\{E_x^B:\; x\in G\}$ is an optimal POVM that decodes $W$ with the lowest probability of error.
\end{itemize}

We have:
\begin{itemize}
\item $\displaystyle\rho^X={\textstyle\Tr_B}(\rho^{XB})=\frac{1}{q}\sum_{x\in G}|x\rangle\langle x|^X$.
\item $\displaystyle\rho^B={\textstyle\Tr_X}(\rho^{XB})=\frac{1}{q}\sum_{x\in G} \rho_x^B$.
\end{itemize}

From \cite[Sec 9.2.3]{Nielsen}, we have
\begin{equation}
D\left(\rho^{XB},\rho^X\otimes\rho^B\right)^2 + F\left(\rho^{XB},\rho^X\otimes\rho^B\right)^2\leq 1,
\label{lalaeqeqvhsbdv}
\end{equation}
where $ D(\rho',\rho'')=\frac{1}{2}\|\rho'-\rho''\|_1$ is the trace distance between $\rho'$ and $\rho''$. We have:

\begin{equation}
\label{eqFidelRhoXBRhoXRhoB}
\begin{aligned}
F\left(\rho^{XB},\rho^X\otimes\rho^B\right)&=\left\|\sqrt{\rho^{XB}}\sqrt{\rho^X\otimes\rho^B}\right\|_1=\left\|\frac{1}{q}\left(\sum_{x\in G}|x\rangle\langle x|^X\otimes \sqrt{\rho_x^B}\right)\cdot\left(\sum_{x\in G}|x\rangle\langle x|^X\otimes \sqrt{\rho^B}\right) \right\|_1\\
&=\frac{1}{q}\left\|\sum_{x\in G}|x\rangle\langle x|^X\otimes \sqrt{\rho_x^B}\sqrt{\rho^B}\right\|_1=\frac{1}{q}\sum_{x\in G}\left\|\sqrt{\rho_x^B}\sqrt{\rho^B}\right\|_1=\frac{1}{q} \sum_{x\in G} F\left(\rho_x^B,\rho^B\right)\\
&= \frac{1}{q} \sum_{x\in G} F\left(\rho_x^B,\frac{1}{q}\sum_{x'\in G}\rho_{x'}^B\right)\stackrel{(a)}{\geq } \frac{1}{q^2} \sum_{x,x'\in G} F\left(\rho_x^B,\rho_{x'}^B\right)=\frac{1}{q^2}\left(q+ \sum_{\substack{x,x'\in G,\\x\neq x'}} F\left(\rho_x^B,\rho_{x'}^B\right)\right)\\
&=\frac{1}{q}\left(1+(q-1)F(W)\right),
\end{aligned}
\end{equation}

where (a) follows from the concavity of the fidelity.

Now let $\mathbb{P}_c(W)=1-\mathbb{P}_e(W)$ be the probability of correct guess of the optimal decoder $\{E_x^B:\; x\in G\}$. We have:

$$\mathbb{P}_c(W)=\frac{1}{q}\sum_{x\in G}\Tr\left(E_x^B\rho_x^B\right)=\frac{1}{q}\sum_{x\in G}\Tr\left(|x\rangle\langle x|^X\otimes E_x^B\rho_x^B\right)= \Tr\left(\Lambda^{XB}\rho^{XB}\right).$$

Therefore,

\begin{align*}
\Tr\left(\Lambda^{XB}\left(\rho^{XB} - \rho^X\otimes\rho^B\right)\right) &= \mathbb{P}_c(W) - \Tr\left(\frac{1}{q}\sum_{x\in G}|x\rangle\langle x|^X\otimes E_x^B\rho^B\right)\\
&=\mathbb{P}_c(W)-\frac{1}{q}\sum_{x\in G}\Tr(E_x^B\rho^B)=\mathbb{P}_c(W)-\frac{1}{q}\stackrel{(a)}{\geq} 0,
\end{align*}
where (a) follows from the fact that a random guess gives a probability of correct guess $\frac{1}{q}$.

On the other hand, we know that $\displaystyle D(\rho^{XB},\rho^X\otimes\rho^B)=\max_{0\leq \Gamma\leq I}\Tr(\Gamma(\rho^{XB}-\rho^X\otimes\rho^B))$. Therefore,
\begin{equation}
\label{eqTraceDistanceLabaYa}
\begin{aligned}
0\leq \mathbb{P}_c(W)-\frac{1}{q} = \Tr\left(\Lambda^{XB}\left(\rho^{XB} - \rho^X\otimes\rho^B\right)\right)&\stackrel{(b)}{\leq} \max_{0\leq \Gamma\leq I}\Tr(\Gamma(\rho^{XB}-\rho^X\otimes\rho^B))\\
&=D\left(\rho^{XB},\rho^X\otimes\rho^B\right),
\end{aligned}
\end{equation}
where (b) follows from the fact that $0\leq \Lambda^{XB}\leq I$.

By combining \eqref{lalaeqeqvhsbdv}, \eqref{eqFidelRhoXBRhoXRhoB} and \eqref{eqTraceDistanceLabaYa}, we get:
$$\left(\mathbb{P}_c(W)-\frac{1}{q}\right)^2 + \frac{1}{q^2}\left(1+(q-1)F(W)\right)^2\leq 1.$$
Thus,
$$\mathbb{P}_c(W)\leq \frac{1}{q}+\sqrt{1-  \frac{1}{q^2}\left(1+(q-1)F(W)\right)^2}=\frac{1+\sqrt{q^2 - \left(1+(q-1)F(W)\right)^2}}{q},$$
which implies that
$$H(X|B)\stackrel{(a)}{\geq} -\log \mathbb{P}_c(W)\geq\log q - \log \left(1+\sqrt{q^2 - \left(1+(q-1)F(W)\right)^2}\right),$$
where (a) follows from \cite[Prop 4.3]{Tomamichel} and the operational interpretation of the conditional min-entropy of a cq-state in terms of the guessing probability \cite{Konig}.
Therefore,
$$I(W)=I(X;B)=H(X)-H(X|B)=\log q - H(X|B) \leq \log \left(1+\sqrt{q^2 - \left(1+(q-1)F(W)\right)^2}\right).$$

\section{Proof of Lemma \ref{lemNonComUnionBound}}
\label{appNonComUnionBound}

Let $\Pi_{r+1}=I$. We have:
\begin{align*}
1-&\Tr\left(\sqrt{\Pi_r}\ldots\sqrt{\Pi_1}\rho \sqrt{\Pi_1}\ldots\sqrt{\Pi_r}\right)\\
&=\Tr\left(\sqrt{\Pi_{r+1}}\rho\sqrt{\Pi_{r+1}}\right)-\Tr\left(\sqrt{\Pi_{r+1}}\ldots\sqrt{\Pi_1}\rho\sqrt{\Pi_1}\ldots\sqrt{\Pi_{r+1}}\right)\\
&=\sum_{i=1}^r \Tr\left(\sqrt{\Pi_{r+1}}\ldots\sqrt{\Pi_{i+1}} \rho\sqrt{\Pi_{i+1}}\ldots\sqrt{\Pi_{r+1}}\right)-\Tr\left(\sqrt{\Pi_{r+1}}\ldots\sqrt{\Pi_i} \rho\sqrt{\Pi_i}\ldots\sqrt{\Pi_{r+1}}\right)\\
&=\sum_{i=1}^r \Tr\left(\sqrt{\Pi_{r+1}}\ldots\sqrt{\Pi_{i+1}} \left(\rho-\sqrt{\Pi_i}\rho\sqrt{\Pi_i}\right)\sqrt{\Pi_{i+1}}\ldots\sqrt{\Pi_{r+1}}\right)\\
&\stackrel{(a)}{\leq}\sum_{i=1}^r \Tr\Bigg(\sqrt{\Pi_{r+1}}\ldots\sqrt{\Pi_{i+1}} \cdot\left|\rho-\sqrt{\Pi_i}\rho\sqrt{\Pi_i}\right|\cdot\sqrt{\Pi_{i+1}}\ldots\sqrt{\Pi_{r+1}}\Bigg)\\
&\stackrel{(b)}{\leq} \sum_{i=1}^r \Tr\left|\rho-\sqrt{\Pi_i}\rho\sqrt{\Pi_i}\right|=\sum_{i=1}^r\left\|\rho-\sqrt{\Pi_i}\rho\sqrt{\Pi_i}\right\|_1\stackrel{(c)}{\leq} 2\sum_{i=1}^r \sqrt{\Tr\left(\rho-\sqrt{\Pi_i}\rho\sqrt{\Pi_i}\right)}\\
&=2r\frac{1}{r} \sum_{i=1}^r \sqrt{1-\Tr(\Pi_i\rho)}\stackrel{(d)}{\leq}2r\sqrt{\frac{1}{r} \sum_{i=1}^r\left(1-\Tr(\Pi_i\rho)\right)}=2\sqrt{r}\sqrt{\sum_{i=1}^r\left(1-\Tr(\Pi_i\rho)\right)},
\end{align*}

where (a) follows from the fact that $\sqrt{\Pi_j}\geq 0$ for every $i+1\leq j\leq r+1$, $\rho-\sqrt{\Pi_i}\rho\sqrt{\Pi_i}\leq \left|\rho-\sqrt{\Pi_i}\rho\sqrt{\Pi_i}\right|$ and the fact that if $A\leq B$ and $C\geq 0$, then $\Tr(AC)\leq \Tr(BC)$. (b) follows from the fact that $0\leq\sqrt{\Pi_j}\leq I$ for every $i+1\leq j\leq r+1$, $\left|\rho-\sqrt{\Pi_i}\rho\sqrt{\Pi_i}\right|\geq 0$, and the fact that if $A,B$ are two positive operators with $B\leq I$, then $\Tr(AB)\leq \Tr(AB)+\Tr(A(I-B))=\Tr(A)$. (c) follows from the fact that $\left\|\rho-\sqrt{X}\rho\sqrt{X}\right\|_1\leq 2\sqrt{\Tr\left(\rho-\sqrt{X}\rho\sqrt{X}\right)}$ for every positive operator $X\leq I$ (see \cite{OgawaNagaoka}). (d) follows from the concavity of the square root.

\section{Proof of Proposition \ref{propFPolar}}
\label{appFPolar}
%

\begin{mylem}
\label{lemTrSqrt}
Let $A$ and $B$ be two positive semi-definite $k\times k$ matrices. We have\footnote{The proof of Lemma \ref{lemTrSqrt} is due to Martin Argerami who thankfully answered our question on Math Stack Exchange. In an earlier version of this paper, we used a weaker inequality $\displaystyle\Tr\sqrt{\sum_{i=1}^n A_i}\leq n\sum_{i=1}^n\Tr\sqrt{A_i}$ which we proved using Weyl's inequality \cite{Weyl} that relates the eigenvalues of $A+B$ with those of $A$ and $B$.}:
$$\Tr\sqrt{A+B}\leq \Tr\sqrt{A}+\Tr\sqrt{B}.$$
\end{mylem}
\begin{proof}
Let us first assume that $A$ and $B$ are invertible. Since the mapping $C\rightarrow C^{-1}$ is monotonically decreasing \cite{Bhatia}, we have $(A+B)^{-1}\leq A^{-1}$. Moreover, since the square root is operator monotone \cite{Bhatia}, we have $(A+B)^{-\frac{1}{2}}\leq A^{-\frac{1}{2}}$. Similarly, $(A+B)^{-\frac{1}{2}}\leq B^{-\frac{1}{2}}$. Therefore,
\begin{align*}
\Tr\sqrt{A+B}&=\Tr\left((A+B)\cdot(A+B)^{-\frac{1}{2}}\right)=\Tr\left(A\cdot(A+B)^{-\frac{1}{2}}\right)+\Tr\left(B\cdot(A+B)^{-\frac{1}{2}}\right)\\
&\stackrel{(a)}{\leq} \Tr\left(A\cdot A^{-\frac{1}{2}}\right)+\Tr\left(B\cdot B^{-\frac{1}{2}}\right)=\Tr\sqrt{A} + \Tr\sqrt{B},
\end{align*}
where (a) follows from the fact that if $C\leq D$ and $A\geq 0$, then $\Tr(AC)\leq \Tr(AD)$.

Now let $A$ and $B$ be two arbitrary positive semi-definite $k\times k$ matrices. We have:
\begin{align*}
\Tr\sqrt{A+B}&=\lim_{\epsilon\to 0} \Tr\sqrt{A+B+2\epsilon I}\leq \lim_{\epsilon\to 0} \Tr\sqrt{A+\epsilon I} + \Tr\sqrt{B+\epsilon I}=\Tr\sqrt{A} + \Tr\sqrt{B}.
\end{align*}
\end{proof}

%

\begin{mylem}
\label{lemFidIneqSum}
Let $\rho_1,\ldots,\rho_n$ and $\sigma_1,\ldots,\sigma_m$ be $n+m$ density matrices of the same dimension. Let $\{p_1,\ldots,p_n\}$ and $\{q_1,\ldots,q_m\}$ be probability distributions on $\{1,\ldots,n\}$ and $\{1,\ldots,m\}$ respectively. We have:
$$F\left(\sum_{i=1}^n p_i\rho_i, \sum_{j=1}^m q_j\sigma_j\right)\leq \sum_{i=1}^n\sum_{j=1}^m \sqrt{p_iq_j}F(\rho_i,\sigma_j).$$
\end{mylem}
\begin{proof}
It is sufficient to show the lemma for the case where $n=1$:
\begin{align*}
F\left(\rho, \sum_{j=1}^m q_j\sigma_j\right)=\Tr\sqrt{\rho^{\frac{1}{2}}\left(\sum_{j=1}^m q_j\sigma_j\right)\rho^{\frac{1}{2}}}\stackrel{(a)}{\leq}  \sum_{j=1}^m \sqrt{q_j}\Tr\sqrt{\rho^{\frac{1}{2}}\sigma_j\rho^{\frac{1}{2}}}=\sum_{j=1}^m \sqrt{q_j}F(\rho,\sigma_j),
\end{align*}
where (a) follows from Lemma \ref{lemTrSqrt}.
\end{proof}

\vspace*{3mm}

Now we are ready to prove Proposition \ref{propFPolar}:

\begin{align*}
F_d(W^+)&=\frac{1}{q}\sum_{x\in G} F(\rho_{x}^+,\rho_{x+d}^+)\\
&=\frac{1}{q}\sum_{x\in G} F\left(\frac{1}{q}\sum_{u_1\in G} \rho_{u_1+x}\otimes\rho_{x}\otimes|u_1\rangle\langle u_1|,\frac{1}{q}\sum_{u_1\in G} \rho_{u_1+x+d}\otimes\rho_{x+d}\otimes|u_1\rangle\langle u_1|\right)\\
&=\frac{1}{q}\sum_{x\in G} F\left(\left(\frac{1}{q}\sum_{u_1\in G}|u_1\rangle\langle u_1|\otimes\rho_{u_1+x}\right)\otimes\rho_{x},\left(\frac{1}{q}\sum_{u_1\in G} |u_1\rangle\langle u_1|\otimes\rho_{u_1+x+d}\right)\otimes\rho_{x+d}\right)\\
&=\frac{1}{q}\sum_{x\in G}
F\left(\frac{1}{q}\sum_{u_1\in G}|u_1\rangle\langle u_1|\otimes\rho_{u_1+x},\frac{1}{q}\sum_{u_1\in G} |u_1\rangle\langle u_1|\otimes\rho_{u_1+x+d}\right)\cdot F\left(\rho_x,\rho_{x+d}\right)\\
&=\frac{1}{q}\sum_{x\in G}
\left(\frac{1}{q}\sum_{u_1\in G}F\left(\rho_{u_1+x},\rho_{u_1+x+d}\right)\right)\cdot F\left(\rho_x,\rho_{x+d}\right)\\
&=\frac{1}{q}\sum_{x\in G}
F_d(W)\cdot F\left(\rho_x,\rho_{x+d}\right)=F_d(W)^2.\\
\end{align*}

\begin{align*}
F_d(W^-)&=\frac{1}{q}\sum_{x\in G} F(\rho_{x}^-,\rho_{x+d}^-)=\frac{1}{q}\sum_{x\in G} F\left(\frac{1}{q}\sum_{u_2\in G} \rho_{x+u_2}\otimes\rho_{u_2},\frac{1}{q}\sum_{u_2\in G} \rho_{x+d+u_2}\otimes\rho_{u_2}\right)\\
&\stackrel{(a)}{\geq}\frac{1}{q^2}\sum_{x,u_2\in G} F(\rho_{x+u_2}\otimes\rho_{u_2}, \rho_{x+d+u_2}\otimes\rho_{u_2})=\frac{1}{q^2}\sum_{x,u_2\in G} F(\rho_{x+u_2}, \rho_{x+d+u_2})=F_d(W),
\end{align*}
where (a) follows from the joint concavity of the fidelity.

\begin{align*}
&F_d(W^-)\\
&=\frac{1}{q}\sum_{x\in G} F\left(\frac{1}{q}\sum_{u_2\in G} \rho_{x+u_2}\otimes\rho_{u_2},\frac{1}{q}\sum_{u_2'\in G} \rho_{x+d+u_2'}\otimes\rho_{u_2'}\right)\\
&\stackrel{(a)}{\leq} \frac{1}{q}\sum_{x\in G}\sum_{u_2,u_2'\in G}\frac{1}{\sqrt{q^2}} F\left(\rho_{x+u_2}\otimes\rho_{u_2},\rho_{x+d+u_2'}\otimes\rho_{u_2'}\right)=\frac{1}{q^2}\sum_{x,u_2,u_2'\in G} F\left(\rho_{x+u_2},\rho_{x+d+u_2'}\right)\cdot F\left(\rho_{u_2},\rho_{u_2'}\right)\\
&=\frac{1}{q^2}\sum_{x,u_2\in G} F\left(\rho_{x+u_2},\rho_{x+d+u_2}\right) + \frac{1}{q^2}\sum_{x,u_2\in G} F\left(\rho_{u_2},\rho_{u_2-d}\right) + \frac{1}{q^2}\sum_{\substack {x,u_2,u_2'\in G,\\
u_2'\neq u_2,\\
u_2'\neq u_2-d}} F\left(\rho_{x+u_2},\rho_{x+d+u_2'}\right)\cdot F\left(\rho_{u_2},\rho_{u_2'}\right)\\
&= 2 F_d(W) + \frac{1}{q^2}\sum_{\substack{\Delta\in G,\\
\Delta\neq 0,\\\Delta\neq -d}}\sum_{x',u_2\in G} F\left(\rho_{x'},\rho_{x'+d+\Delta}\right)F\left(\rho_{u_2},\rho_{u_2+\Delta}\right)\\
&=2 F_d(W) + \sum_{\substack{\Delta\in G,\\
\Delta\neq 0,\\\Delta\neq -d}} F_{\Delta}(W)F_{d+\Delta}(W),
\end{align*}
where (a) follows from Lemma \ref{lemFidIneqSum}.

\section{Proof of Lemma \ref{lemFMDFmax}}
\label{appFMDFmax}

\begin{align*}
F(W[M|D])&= \frac{1}{|D/M|(|D/M|-1)}\sum_{\substack{C,C'\in D/M,\\C\neq C'}}F(\rho_C,\rho_{C'})\\
&=\frac{|M|^2}{|H|(|H|-|M|)}\sum_{\substack{C,C'\in D/M,\\C\neq C'}}F\left(\frac{1}{|C|}\sum_{x\in C} \rho_x,\frac{1}{|C'|}\sum_{x'\in C'}\rho_{x'}\right)\\
&\stackrel{(a)}{\leq} \frac{|M|^2}{|H|(|H|-|M|)\sqrt{|C|\cdot|C'|}}\sum_{\substack{C,C'\in D/M,\\C\neq C'}}\sum_{\substack{x\in C,\\x'\in C'}}F(\rho_x,\rho_{x'})\\
&\stackrel{(b)}{\leq}\frac{|M|}{|H|(|H|-|M|)}\sum_{\substack{x\in D,\\d\in H,\\d\notin M}} F(\rho_x,\rho_{x+d})\leq \frac{|M|}{|H|(|H|-|M|)}\sum_{\substack{d\in H,\\d\notin M}}\frac{q}{q}\sum_{x\in G} F(\rho_x,\rho_{x+d})\\
&=\frac{q\cdot|M|}{|H|(|H|-|M|)}\sum_{\substack{d\in H,\\d\notin M}}F_d(W)\leq \frac{q\cdot|M|}{|H|(|H|-|M|)}(|H|-|M|)F_{\max}^{M|H}(W),
\end{align*}
where (a) follows from Lemma \ref{lemFidIneqSum}, and (b) follows from the fact that $|C|=|C'|=|M|$ and the fact that \Big\{$\exists C,C'\in D/M$: $x\in C$, $x'\in C'$ and $C\neq C'$\Big\} if and only if \Big\{$x\in D$, $x'-x\in H$ and $x'-x\notin M$\Big\}.

Now let us show the second inequality of Lemma \ref{lemFMDFmax}. Assume that $M$ is maximal in $H$ and let $d\in H$ be such that $d\notin M$ and $F_{\max}^{M|H}(W)=F_d(W)$. Since $\displaystyle 1-F_d(W)=\frac{1}{q}\sum_{x\in G}\big(1-F(\rho_x,\rho_{x+d})\big)$, we have $F(\rho_x,\rho_{x+d})\geq 1-q(1-F_d(W))=1-q\big(1-F_{\max}^{M|H}(W)\big)$ for every $x\in G$.

For every $C\in D/M$, we have:
\begin{align*}
F(\rho_C,\rho_{d+C}) &\stackrel{(a)}{\geq} 1- D(\rho_C,\rho_{d+C})= 1-D\left (\frac{1}{|C|}\sum_{x\in C}\rho_x,\frac{1}{|C|}\sum_{x\in C}\rho_{x+d}\right)=1-\frac{1}{2}\left\|\frac{1}{|C|}\sum_{x\in C}(\rho_x-\rho_{x+d})\right\|_1\\
&\geq
1-\frac{1}{|C|}\sum_{x\in C}\frac{1}{2}\left\|\rho_x-\rho_{x+d}\right\|_1= 1-\frac{1}{|C|}\sum_{x\in C}  D(\rho_x,\rho_{x+d})\\
&\stackrel{(b)}{\geq} 1-\frac{1}{|C|}\sum_{x\in C}\sqrt{1-F(\rho_x,\rho_{x+d})^2}\geq 1-\sqrt{1-\Big(1-q\big(1-F_{\max}^{M|H}(W)\big)\Big)^2},
\end{align*}
where (a) follows from the fact that $ D(\rho',\rho'')+F(\rho',\rho'')\geq 1$ (see \cite{Nielsen}). (here $ D(\rho',\rho'')=\frac{1}{2}\|\rho'-\rho''\|_1$ is the trace distance between $\rho'$ and $\rho''$.) (b) follows from the fact that $ D(\rho',\rho'')^2+F(\rho',\rho'')^2\leq 1$ (see \cite{Nielsen}).

Now let $C,C'\in D/M$ be such that $C\neq C'$. Since $|H/M|$ is prime, we can write $C'=ld+C$ for some $0\leq l<|H/M|$. We have:

\begin{align*}
F(\rho_C,\rho_{C'})&=F(\rho_C,\rho_{ld+C})=\cos A(\rho_C,\rho_{ld+C})\stackrel{(a)}{\geq} \cos\left(\sum_{i=0}^{l-1}A\big(\rho_{id+C},\rho_{(i+1)d+C}\big)\right)\\
&=\cos\left(\sum_{i=0}^{l-1} \arccos F\big(\rho_{id+C},\rho_{(i+1)d+C}\big)\right)\\
&\stackrel{(b)}{\geq} \cos\left(l\cdot \arccos \left(1-\sqrt{1-\Big(1-q\big(1-F_{\max}^{M|H}(W)\big)\Big)^2} \right)\right)\\
&\stackrel{(c)}{\geq} \cos\left(\frac{|H|-|M|}{|M|} \arccos \left(1-\sqrt{1-\Big(1-q\big(1-F_{\max}^{M|H}(W)\big)\Big)^2} \right)\right),
\end{align*}
where (a) follows from the fact that $A(\rho',\rho'')=\arccos F(\rho',\rho'')$ is a metric \cite{Nielsen}. Note that since $\cos$ is a decreasing function on $\displaystyle \left[0,\frac{\pi}{2}\right]$, (a), (b) and (c) become true if we assume that $\displaystyle 1-\sqrt{1-\Big(1-q\big(1-F_{\max}^{M|H}(W)\big)\Big)^2}\geq \cos\left(\frac{\pi}{2(q-1)}\right)$. In other words, we can take $$\delta_q = \frac{1}{q}\left(1-\sqrt{1-\left(1-\cos\left(\frac{\pi}{2(q-1)}\right)\right)^2}\right).$$

We conclude that
\begin{align*}
F(W[M|D])&=\frac{1}{|D/M|(|D/M|-1)}\sum_{\substack{C,C'\in D/M,\\C\neq C'}}F(\rho_C,\rho_{C'})\\
&\geq \cos\left(\frac{|H|-|M|}{|M|} \arccos \left(1-\sqrt{1-\Big(1-q\big(1-F_{\max}^{M|H}(W)\big)\Big)^2} \right)\right).
\end{align*}

\section{Proof of Lemma \ref{lemFWHPolar}}

\label{appFWHPolar}

\begin{mylem}
\label{lemFMaxLeqFH}
For every subgroup $H$ of $G$, we have:
$$F_{\max}^{H|G}(W)\leq (q-|H|) F(W[H])$$
\end{mylem}
\begin{proof}
\begin{align*}
F(W[H])&=\frac{1}{|G/H|(|G/H|-1)} \sum_{\substack{C,C'\in G/H,\\C\neq C'}}F(\rho_C,\rho_{C'})\\
&=\frac{1}{|G/H|(|G/H|-1)} \sum_{\substack{C,C'\in G/H,\\C\neq C'}}F\left(\frac{1}{|C|} \sum_{x\in C}\rho_x,\frac{1}{|C'|}\sum_{x'\in C'} \rho_{x'}\right)\\
&\stackrel{(a)}{\geq} \frac{1}{|G/H|(|G/H|-1)}\cdot\frac{1}{|H|^2} \sum_{\substack{C,C'\in G/H,\\C\neq C'}}  \sum_{\substack{x\in C,\\x'\in C'}} F(\rho_x,\rho_{x'})\\
&=\frac{1}{q(q-|H|)} \sum_{\substack{x,d\in G,\\d\notin H}} F(\rho_x,\rho_{x+d})=\frac{1}{q-|H|}\sum_{\substack{d\in G,\\d\notin H}} F_d(W)\geq \frac{1}{q-|H|} F_{\max}^{H|G}(W),
\end{align*}
where (a) follows from the concavity of the fidelity and from the fact that $|C|=|C'|=|H|$.
\end{proof}

\vspace*{3mm}

Now we are ready to prove Lemma \ref{lemFWHPolar}. The lemma is trivial for $H=G$. Assume that $H\neq G$. We have:
\begin{align*}
F(W^-[H])&=F(W^-[H|G]) \stackrel{(a)}{\leq} \frac{q\cdot |H|}{q} F_{\max}^{H|G}(W^-)= |H| \max_{\substack{d\in G,\\d\notin H}}F_d(W^-)\\
&\stackrel{(b)}{\leq} |H| \max_{\substack{d\in G,\\d\notin H}}\bigg\{2 F_d(W) + \sum_{\substack{\Delta\in G,\\
\Delta\neq 0,\\\Delta\neq -d}}F_\Delta(W)F_{d+\Delta}(W)\bigg\}\\
&\stackrel{(c)}{\leq} |H| \bigg(2 F_{\max}^{H|G}(W) + (q-2)F_{\max}^{H|G}(W)\bigg)= |H|qF_{\max}^{H|G}(W)\\
&\stackrel{(d)}{\leq} |H|q(q-|H|) F(W[H]),
\end{align*}
where (a) follows from Lemma \ref{lemFMDFmax}. (b) follows from Proposition \ref{propFPolar}. (c) follows from the fact that for every $d,\Delta\in G$, if $d\notin H$ then either $\Delta\notin H$ or $d+\Delta\notin H$, and so $F_\Delta(W)F_{d+\Delta}(W)\leq F_{\max}^{H|G}(W)$. (d) follows from Lemma \ref{lemFMaxLeqFH}.

On the other hand,
\begin{align*}
F(W^+[H])&=F(W^+[H|G]) \stackrel{(a)}{\leq} \frac{q\cdot|H|}{q} F_{\max}^{H|G}(W^+)= |H| \max_{\substack{d\in G,\\d\notin H}}F_d(W^+)\\
&\stackrel{(b)}{=} |H| \max_{\substack{d\in G,\\d\notin H}}F_d(W)^2= |H| F_{\max}^{H|G}(W)^2 \stackrel{(c)}{\leq} |H| (q-|H|)^2 F(W[H])^2,
\end{align*}
where (a) follows from Lemma \ref{lemFMDFmax}, (b) follows from Proposition \ref{propFPolar} and (c) follows from Lemma \ref{lemFMaxLeqFH}.

\section{Proof of Lemma \ref{lemPOVMFWH}}
\label{appPOVMFWH}
It is sufficient to show the following simpler version:

\begin{mylem}
\label{lemPOVMFWHSimple}
If $W: x\in G\longrightarrow \rho_x\in\mathcal{DM}(k\cdot r)$ is a cq-channel such that
$$\rho_x^{BU}=\frac{1}{r}\sum_{u=1}^r \rho_{x,u}^B\otimes|u\rangle\langle u|^U,$$
where $\rho_{x,u}^B\in\mathcal{DM}(k)$ and $\{|u\rangle^U:\;1\leq u\leq r\}$ is an orthonormal basis of the Hilbert space of dimension $r$, then for every $1\leq u\leq r$, there exists a POVM $\left\{\Pi_{u,x}^B:\;x\in G \right\}$ such that the POVM $\left\{\Pi_x^{BU}:\;x\in G \right\}$ defined as
$$\Pi_x^{BU}=\sum_{u=1}^r \Pi_{u,x}^B\otimes|u\rangle\langle u|^U,$$
satisfies
$$1-\frac{1}{q}\sum_{x\in G}\Tr\left(\Pi_x^{BU} \rho_x^{BU}\right)<(q-1)F(W).$$
\end{mylem}
\begin{proof}
For every $1\leq u\leq r$, define the cq-channel $W_u:x\in G\longrightarrow \rho_{x,u}\in\mathcal{DM}(k)$. The optimal decoder for $W_u$ satisfies $\mathbb{P}_e(W_u)\leq (q-1)F(W_u)$ \cite{AvFidelity}. Therefore, there exists a POVM $\left\{\Pi_{u,x}^B:\;x\in G \right\}$ satisfying,
$$1-\frac{1}{q}\sum_{x\in G}\Tr\left(\Pi_{u,x}^{B} \rho_{u,x}^{B}\right)<(q-1)F(W_u).$$

For every $x\in G$, define
$$\Pi_x^{BU}=\sum_{u=1}^r \Pi_{u,x}^B\otimes|u\rangle\langle u|^U.$$
It is easy to see that $\left\{\Pi_x^{BU}:\;x\in G \right\}$ is a valid POVM. We have:
\begin{align*}
1-\frac{1}{q}\sum_{x\in G}\Tr\left(\Pi_x^{BU} \rho_x^{BU}\right)&=
1-\frac{1}{qr}\sum_{x\in G}\sum_{u=1}^r\Tr\left(\Pi_{u,x}^{B} \rho_{u,x}^{B}\right)=\frac{1}{r}\sum_{u=1}^r\left(1-\frac{1}{q}\sum_{x\in G}\Tr\left(\Pi_{u,x}^{B} \rho_{u,x}^{B}\right)\right)\\
&\leq \frac{1}{r}\sum_{u=1}^r (q-1)F(W_u)=\frac{q-1}{r}\sum_{u=1}^r \sum_{\substack{x,x'\in G,\\x\neq x'}}F(\rho_{u,x}^B,\rho_{u,x'}^B)\\
&=(q-1)\sum_{\substack{x,x'\in G,\\x\neq x'}}F\left(\frac{1}{r}\sum_{u=1}^r\rho_{x,u}^B\otimes|u\rangle\langle u|^U,\frac{1}{r}\sum_{u=1}^r\rho_{x',u}^B\otimes|u\rangle\langle u|^U \right)\\
&=(q-1)F(W).
\end{align*}
\end{proof}

\bibliographystyle{IEEEtran}
\bibliography{bibliofile}

\begin{thebibliography}{10}
\providecommand{\url}[1]{#1}
\csname url@samestyle\endcsname
\providecommand{\newblock}{\relax}
\providecommand{\bibinfo}[2]{#2}
\providecommand{\BIBentrySTDinterwordspacing}{\spaceskip=0pt\relax}
\providecommand{\BIBentryALTinterwordstretchfactor}{4}
\providecommand{\BIBentryALTinterwordspacing}{\spaceskip=\fontdimen2\font plus
\BIBentryALTinterwordstretchfactor\fontdimen3\font minus
  \fontdimen4\font\relax}
\providecommand{\BIBforeignlanguage}[2]{{%
\expandafter\ifx\csname l@#1\endcsname\relax
\typeout{** WARNING: IEEEtran.bst: No hyphenation pattern has been}%
\typeout{** loaded for the language `#1'. Using the pattern for}%
\typeout{** the default language instead.}%
\else
\language=\csname l@#1\endcsname
\fi
#2}}
\providecommand{\BIBdecl}{\relax}
\BIBdecl

\bibitem{Arikan}
E.~Ar{\i}kan, ``Channel polarization: A method for constructing
  capacity-achieving codes for symmetric binary-input memoryless channels,''
  \emph{Information Theory, IEEE Transactions on}, vol.~55, no.~7, pp. 3051
  --3073, 2009.

\bibitem{ArikanTelatar}
E.~Ar{\i}kan and E.~Telatar, ``On the rate of channel polarization,'' in
  \emph{Information Theory, 2009. ISIT 2009. IEEE International Symposium on},
  28 2009.

\bibitem{SasogluTelAri}
E.~\c{S}a\c{s}o\u{g}lu, E.~Telatar, and E.~Ar{\i}kan, ``Polarization for
  arbitrary discrete memoryless channels,'' in \emph{Information Theory
  Workshop, 2009. ITW 2009. IEEE}, 2009, pp. 144 --148.

\bibitem{ParkBarg}
W.~Park and A.~Barg, ``Polar codes for $q$-ary channels,,'' \emph{Information
  Theory, IEEE Transactions on}, vol.~59, no.~2, pp. 955--969, 2013.

\bibitem{SahebiPradhan}
A.~G. Sahebi and S.~S. Pradhan, ``Multilevel channel polarization for arbitrary
  discrete memoryless channels,'' \emph{IEEE Transactions on Information
  Theory}, vol.~59, no.~12, pp. 7839--7857, Dec 2013.

\bibitem{RajTelA}
R.~Nasser and E.~Telatar, ``Polar codes for arbitrary dmcs and arbitrary
  macs,'' \emph{IEEE Transactions on Information Theory}, vol.~62, no.~6, pp.
  2917--2936, June 2016.

\bibitem{RajErgI}
\BIBentryALTinterwordspacing
R.~Nasser, ``Ergodic theory meets polarization. {I}: An ergodic theory for
  binary operations,'' \emph{CoRR}, vol. abs/1406.2943, 2014. [Online].
  Available: \url{http://arxiv.org/abs/1406.2943}
\BIBentrySTDinterwordspacing

\bibitem{RajErgII}
\BIBentryALTinterwordspacing
------, ``Ergodic theory meets polarization. {I}{I}: A foundation of
  polarization theory,'' \emph{CoRR}, vol. abs/1406.2949, 2014. [Online].
  Available: \url{http://arxiv.org/abs/1406.2949}
\BIBentrySTDinterwordspacing

\bibitem{SasogluTelYeh}
E.~\c{S}a\c{s}o\u{g}lu, E.~Telatar, and E.~M. Yeh, ``Polar codes for the
  two-user multiple-access channel,'' \emph{IEEE Transactions on Information
  Theory}, vol.~59, no.~10, pp. 6583--6592, Oct 2013.

\bibitem{AbbeTelatar}
E.~Abbe and E.~Telatar, ``Polar codes for the -user multiple access channel,''
  \emph{Information Theory, IEEE Transactions on}, vol.~58, no.~8, pp. 5437
  --5448, aug. 2012.

\bibitem{WildeGuhaCQ}
M.~M. Wilde and S.~Guha, ``Polar codes for classical-quantum channels,''
  \emph{IEEE Transactions on Information Theory}, vol.~59, no.~2, pp.
  1175--1187, Feb 2013.

\bibitem{PolarcqMAC}
C.~Hirche, C.~Morgan, and M.~M. Wilde, ``Polar codes in network quantum
  information theory,'' \emph{IEEE Transactions on Information Theory},
  vol.~62, no.~2, pp. 915--924, Feb 2016.

\bibitem{Monotone}
E.~Arikan, ``Polar coding for the slepian-wolf problem based on monotone chain
  rules,'' in \emph{Information Theory Proceedings (ISIT), 2012 IEEE
  International Symposium on}, July 2012, pp. 566--570.

\bibitem{AvFidelity}
\BIBentryALTinterwordspacing
H.~Barnum and E.~Knill, ``Reversing quantum dynamics with near-optimal quantum
  and classical fidelity,'' \emph{Journal of Mathematical Physics}, vol.~43,
  no.~5, pp. 2097--2106, 2002. [Online]. Available:
  \url{http://aip.scitation.org/doi/abs/10.1063/1.1459754}
\BIBentrySTDinterwordspacing

\bibitem{Sen}
P.~Sen, ``Achieving the han-kobayashi inner bound for the quantum interference
  channel by sequential decoding,'' \emph{arXiv:1109.0802}, September 2011.

\bibitem{Nielsen}
M.~A. Nielsen and I.~L. Chuang, \emph{Quantum Computation and Quantum
  Information: 10th Anniversary Edition}, 10th~ed.\hskip 1em plus 0.5em minus
  0.4em\relax New York, NY, USA: Cambridge University Press, 2011.

\bibitem{SasogluThesis}
\BIBentryALTinterwordspacing
E.~\c{S}a\c{s}o\u{g}lu, ``Polar {C}oding {T}heorems for {D}iscrete {S}ystems,''
  Ph.D. dissertation, IC, Lausanne, 2011. [Online]. Available:
  \url{http://library.epfl.ch/theses/?nr=5219}
\BIBentrySTDinterwordspacing

\bibitem{cqMAC}
A.~Winter, ``The capacity of the quantum multiple-access channel,'' \emph{IEEE
  Transactions on Information Theory}, vol.~47, no.~7, pp. 3059--3065, Nov
  2001.

\bibitem{RajTelFourier}
R.~Nasser and E.~Telatar, ``Fourier analysis of mac polarization,''
  \emph{arXiv:1501.06076}, January 2015.

\bibitem{HolevoRel}
A.~S. Holevo, ``Reliability function of general classical-quantum channel,''
  \emph{IEEE Transactions on Information Theory}, vol.~46, no.~6, pp.
  2256--2261, Sep 2000.

\bibitem{Tomamichel}
M.~Tomamichel, ``A framework for non-asymptotic quantum information theory,''
  \emph{arXiv:1406.2943}, 2012.

\bibitem{Konig}
R.~Konig, R.~Renner, and C.~Schaffner, ``The operational meaning of min- and
  max-entropy,'' \emph{IEEE Transactions on Information Theory}, vol.~55,
  no.~9, pp. 4337--4347, Sept 2009.

\bibitem{OgawaNagaoka}
T.~Ogawa and H.~Nagaoka, ``Making good codes for classical-quantum channel
  coding via quantum hypothesis testing,'' \emph{IEEE Transactions on
  Information Theory}, vol.~53, no.~6, pp. 2261--2266, June 2007.

\bibitem{Weyl}
\BIBentryALTinterwordspacing
H.~Weyl, ``Das asymptotische verteilungsgesetz der eigenwerte linearer
  partieller differentialgleichungen (mit einer anwendung auf die theorie der
  hohlraumstrahlung),'' \emph{Mathematische Annalen}, vol.~71, pp. 441--479,
  1912. [Online]. Available: \url{http://eudml.org/doc/158545}
\BIBentrySTDinterwordspacing

\bibitem{Bhatia}
R.~Bhatia, \emph{Positive Definite Matrices}, ser. Princeton Series in Applied
  Mathematics.\hskip 1em plus 0.5em minus 0.4em\relax Princeton University
  Press, 2009.

\end{thebibliography}
\end{document}